\documentclass[final,12pt]{article}

\usepackage{chngcntr}

\usepackage{amsfonts}
\usepackage{eurosym}
\usepackage{setspace}
\usepackage{color}
\usepackage{lscape}
\usepackage{ifthen}
\usepackage{amssymb}
\usepackage{amsmath}
\usepackage{amsthm}
\usepackage{enumerate}
\usepackage{fancyhdr}
\usepackage{geometry}
\usepackage{setspace}
\usepackage{abstract}
\usepackage{natbib}
\usepackage{url}
\usepackage{rotating}
\usepackage{overpic}
\usepackage{pstricks}
\usepackage{pst-plot}
\usepackage{hyperref}
\usepackage{tabularx,booktabs}
\usepackage{multirow,tabularx}
\usepackage{xcolor}
\usepackage{titlesec}
\usepackage{color}
\usepackage{enumitem}
\usepackage{booktabs}
\usepackage{algorithm}
\usepackage{algorithmic}
\usepackage{subcaption}
\usepackage{chngcntr}
\usepackage[section]{placeins}

\newtheorem{proposition}{Proposition}

\providecommand{\U}[1]{\protect\rule{.1in}{.1in}}
\newcolumntype{Y}{>{\centering\arraybackslash}X}
\titleclass{\subsubsubsection}{straight}[\subsubsection]
\newcounter{subsubsubsection}
\renewcommand{\thesubsubsubsection}{\arabic{subsubsubsection}}
\definecolor{darkblue}{rgb}{0,0,0.55}
\definecolor{darkred}{rgb}{0.5,0,0}
\newcommand\Bheadfont{\fontsize{14pt}{\baselineskip}\selectfont}
\titleformat{\section}[hang] {\normalfont\sc\color{darkblue}\Bheadfont} {\thesection\hskip0.618em}{0em}{}
\titlespacing*{\section}
{0pt}{15pt plus 2pt minus 2pt}{9pt plus 2pt minus 2pt}
\titleformat{\subsection}[runin]
{\normalfont\sc\color{darkblue}} {\thesubsection\hskip0.618em}{0em}{}
\titlespacing*{\subsection}
{0pt}{13pt plus 2pt minus 2pt}{13pt plus 2pt minus 2pt}
\titleformat{\subsubsection}[runin]
{\normalfont\sc\color{darkblue}} {\thesubsubsection\hskip0.618em}{0em}{}
\titlespacing*{\subsubsection}
{0pt}{13pt plus 2pt minus 2pt}{13pt plus 2pt minus 2pt}
\titleformat{\subsubsubsection}[runin]
{\normalfont\sc\color{darkblue}} {\thesubsubsubsection\hskip0.618em}{0em}{}
\titlespacing*{\subsubsubsection}
{0pt}{13pt plus 2pt minus 2pt}{13pt plus 2pt minus 2pt}
\newcommand{\bbm}{\begin{bmatrix}}
\newcommand{\ebm}{\end{bmatrix}}

\newcommand{\bd}{\begin{description}}
\newcommand{\ed}{\end{description}}

\title{Deep Reinforcement Learning in a Monetary Model}

\author{Mingli Chen\thanks{University of Warwick} \and Andreas Joseph\thanks{Bank of England. {\it Disclaimer:} The views expressed in this work do not necessarily represent those of the Bank of England or its committees.} \and Michael Kumhof\footnotemark[2] \and Xinlei Pan\thanks{University of California, Berkeley}   \and Xuan Zhou\thanks{Deakin University\newline
The support of the Economic and Social Research Council (ESRC) is gratefully acknowledged, via the Rebuilding Macroeconomics Network (Grant Ref: ES/R00787X/1). We would like to especially thank Rui Shi for her contribution in the early stages of the project. We thank Angus Armstrong, Roger Farmer, Ekaterina Svetlova, and Yaolang Zhong for useful comments. All errors are ours.}}

\date{ \today}

\begin{document}

\maketitle
 
\begin{abstract}
We propose using deep reinforcement learning to solve dynamic stochastic general equilibrium models. Agents are represented by deep artificial neural networks and learn to solve their dynamic optimisation problem by interacting with the model environment, of which they have no a priori knowledge. Deep reinforcement learning offers a flexible yet principled way to model bounded rationality within this general class of models.  We apply our proposed approach to a classical model from the adaptive learning literature in macroeconomics which looks at the interaction of monetary and fiscal policy. We find that, contrary to adaptive learning, the artificially intelligent household can solve the model in all policy regimes. 
 
\end{abstract}
\bigskip
\noindent \textbf{Key Words}:  Artificial Intelligence, Deep Reinforcement Learning, Adaptive Learning, Monetary Policy, Fiscal Policy.
 \medskip{}

\noindent \textbf{JEL Codes: } C14, C52, D83, E52, E62 \medskip{}
\newpage

\section{Introduction}
\label{sec:intro}

Agent expectations are central to macroeconomics.  The idea of rational expectations assumes individual rationality and consistency of expectations for all the agents in the model, and when implemented numerically or econometrically, rational expectations models impute much more knowledge to the agents within the model than is possessed by an econometrician \citep{sargent1993bounded, evans2009learning}. While the rational expectations benchmark is a natural one to consider, the assumptions underpinning it are strong, and one may wonder if they should be relaxed \citep{woodford2013macroeconomic}.

The literature on Adaptive Learning \citep{sargent1993bounded, EH2001book}, one of the leading paradigms in the learning literature and the one that we take as our starting and reference point in this paper, retains the assumption of individual rationality, while replacing the assumption of consistency of expectations with the assumption that agents form their expectations adaptively, and use recursive linear least squares as a forecasting rule. These forecasts are an input into agent's decision rules, and in each period the economy attains a temporary equilibrium. Models populated with \textit{Adaptively Learning Agents} put the agents on an equal footing with the econometrician who is observing data from the model. However, this type of parametric recursive method assumes that agents correctly specify the laws of motion and other relevant functional relationships of the model. By assumption, the predictions of this econometric model need not coincide with the predictions of the true model. It is, therefore, important to correctly specify the reduced form forecasting rule such that the learning agent's expectations converge to the rational ones (in this case, an equilibrium is referred to as learnable). Moreover, economic dynamics, e.g. the stability of central bank policies such as Taylor rules or forward guidance, may be different under adaptively learning agents compared to fully rational ones \citep{EP2018imperfect}.

In this paper, we propose to combine a standard dynamic stochastic general equilibrium modeling approach with flexible expectations formation, by using modern developments in deep reinforcement learning \citep{mnih2015human}. Specifically, we work with models populated by \textit{Deep Reinforcement Learning Agents} (a.k.a. \textit{Artificially Intelligent Agents}) who have no a priori knowledge about the structure of the economy, and, instead, use their utility realisations in response to their actions in order to learn nonlinear decision rules via deep artificial neural networks \citep{goodfellow2016deep}. Artificially Intelligent Agents can be trained to learn a good strategy to apply within an environment that is complicated \citep{Sutton1998RL}. We adopted a policy-based deep reinforcement learning approach that can deal with high dimensional continuous action spaces \citep{Haarnoja2018sac}. Our approach enables agents to learn flexibly, as our learning algorithms are nonparametric and recursive, therefore reducing the risk of misspecification induced by the parametric approach. Allowing for misspecification and learning via expelling rational expectation agents and replacing them with “artificially intelligent”  ones is also reminiscent of the paradigm of \cite{sargent1993bounded} populating models with boundedly rational agents via introducing artificially intelligent agents.  

Reinforcement learning is about understanding how agents learn and make optimal decisions through repeated experience \citep{Sutton1998RL}. Agents strive to maximize some long-term reward, similar to the cumulated discounted sum of future utilities as in classical macroeconomic models, by interacting with a generally unknown environment. The environment in our case is the model economy. Agents obtain observations of state variables from this environment. They take actions based on these observations, e.g. how much to consume this period. They are then rewarded by the environment, which also returns a new state observation, and so on. The two main differences to adaptive learning are that agents in reinforcement learning have very limited information about the economy, and that neither their behaviour nor their expectation process is explicitly formulated. Thus, inference happens from their reward function together with the transitions generated by interacting with the environment.\footnote{Reinforcement learning is related to the value function iteration procedure which has been used in the context of dynamics programming in economics.}

The recent advances of deep reinforcement learning -- reinforcement learning with function approximation by deep artificial neural networks -- have improved the performance of traditional reinforcement learning in several very challenging domains such as computer games and simulated robotics in the computer science and machine learning community \citep{mnih2013atari,mnih2015human}. The use of deep artificial neural networks, which belongs to the class of universal function approximators \citep{Cybenko1989,goodfellow2016deep}, in reinforcement learning reduces the risk of misspecification and is, at the same time, at the forefront of advances in artificial intelligence, where agents learn to master complex dynamic environments. We contribute to the literature by showing that deep reinforcement learning can be used by economists to solve complex behavioral problems, and that this approach holds promise for modelling expectations in economic models. 

We start with investigating how our proposed deep reinforcement learning approach enables agents to learn ``sensible'' equilibria, and comparing with results from Adaptively Learning Agents. We apply our approach to a classical model from the learning literature in macroeconomics, which looks at the interaction of monetary and fiscal policy with a single representative household agent. The model considers inflation and debt dynamics under a global Taylor-rule. Given the zero lower bound on interest rates, a global Taylor-rule is known to generate two steady states of inflation steady state. One is the inflation target, while the other is a low inflation ``liquidity trap'', where is a continuum of perfect foresight paths that start around the inflation target and converge to the low inflation steady state.\footnote{More detailed discussions can be found in \cite{BSGU2001a,BSGU2001b} among others.}  \cite{EH2005,Eusepi2007,EH2008} have studied the E-stability properties of these two steady states under adaptive learning, and show that the learnability of the two steady states depends on the specifications of monetary and fiscal policy.
 
We find that when an active fiscal or monetary policy is paired with a passive policy, i.e. active monetary policy and passive fiscal policy, or, passive monetary policy and active fiscal policy, the corresponding rational expectations equilibria are determinate, and both Adaptively Learning Agent and Artificially Intelligent Agent can learn these equilibria. However, when both monetary and fiscal policy are active (or passive), the corresponding rational expectations equilibria are explosive (or indeterminate), and contrary to the Adaptively Learning Agent, our Artificially Intelligent Agent remains capable of learning these equilibria. Our results also echo some early results in the literature comparing adaptive learning with artificial agents, e.g. learning via genetic algorithms \citep{arifovic1995genetic}. The main reason why all regimes are learnable by the Artificially Intelligent Agent as compared to the Adaptively Learning Agent is that the former is not constrained by the dynamics of the linearised system, but instead by the ``global map'' given by long-term utility maximization. The more general results obtained from our Artificially Intelligent Agent also means that the economy can end up in potentially more states than previously thought.

During learning, we assess the state of learning or the state of (bounded) rationality of our agents. The Artificially Intelligent Agent's learning can generally be characterised by three phases: an initial {\it random phase} due to algorithm initialisation, followed by a {\it learning phase}, and finally the achievement of the rational expectations equilibrium, or a {\it rational phase}. We define a set of measures to quantify this process, which we label first-order condition distances and steady state distances for learning about behaviour and state values, respectively. The agent's behaviour and learned state values constitute solutions of the model during each of the three phases, though suboptimal before reaching the rational phase.

Finally, we show how the implicit expectations of the Artificially Intelligent Agent can be extracted in a small numerical experiment that looks into the process of learning about inflation expectations. Together with the proposed measures of bounded rationality, this type of analysis may help to eventually bring the proposed approaches to the data. 

The remainder of this paper is structured as follows. Section \ref{sec:model} introduces the model. Section \ref{sec:learning} provides introductions to both adaptive and (deep) reinforcement learning. Section\ \ref{sec:results} presents the main results. We conclude with a general discussion in Section\ \ref{sec:conclusion}. Auxiliary information is provided in the \hyperref[sec:app]{Appendix}.

\section{Model}
\label{sec:model}

The model closely follows \cite{BSGU2001b} and \cite{EH2005}. Time is discrete and prices are flexible.

\subsection{Household}
There is a single representative household who discounts the future at a rate $\beta \in (0, 1)$. The agent has access to fiat money and nominal government bonds. Agents seek to maximize their utility, which depends on consumption, real money balances and hours worked, subject to an inter-temporal budget constraint. Formally, the household solves the following problem 

\begin{equation}
\label{eq:value}
\max_{c_t,m_t,n_t}\ \ \ \ \mathbb{E}_{0}\sum_{t=0}^{\infty }\beta ^{t} U(c_t,m_t,n_t)  \; \; \text{s.t.}
\end{equation}

\begin{equation}
M_t+B_{t}+P_{t}c_{t}=M_{t-1}+B_{t-1}R_{t-1}+ W_tn_t-P_t\tau_t ,
\end{equation}
where $P_t$ is the price level at time $t$, $x_t=\frac{X_t}{P_t},\,X_t \in\{M_t,B_t,C_t,W_t\}$. $x_t$ denote the the real levels of money, government bonds, consumption and wages. $n_t$ is the hours worked. The household pays $\tau_t$ as a real lump-sum tax to the government each period.

\subsection{Optimality conditions}
The agent chooses sequences of money, bonds, labour supply and consumption, taking the price of goods, the real wage and nominal interest rates as given. The optimality conditions are given by
 \begin{align}
&\text{Euler Equation:\ \ \ \ }  U_{c,t}=\beta E_t {U_{c, t+1}} \frac{R_t}{\pi_{t+1}} \label{eq:euler}\\
&\text{Money Demand:\ \ \ } \frac{U_{m,t}}{U_{c,t}}=\frac{R_t-1}{R_t} \label{eq:md}\\
&\text{Labor Supply:\ \ \ } -\frac{U_{n,t}}{U_{c,t}}=w_t \label{eq:ls}
\end{align}
where $R_t$ is the nominal interest rate on government bonds, and $\pi_t$ is the inflation rate at time t. 
%
We follow \cite{EH2005} by adopting a utility function of the form \footnote{We assume the utility function is separable for simplicity. Our main results do not change if non-separable function is applied, see \cite{BSGU2001a} and \cite{Eusepi2007} for details.}
 \begin{align}
U(c_t, m_t,n_t)=\frac{c_t^{1-\sigma}}{1-\sigma}+\chi \frac{m_t^{1-\sigma}}{1-\sigma}- \frac{n_t^{1+\varphi}}{1+\varphi}.\label{eq:utility}
\end{align}
Combining the utility \eqref{eq:utility} and the household optimality conditions \eqref{eq:euler} -- \eqref{eq:ls}, we arrive the Euler equation
 \begin{align}
1=\beta E_t (\frac{c_{t+1}}{c_t})^{-\sigma}\frac{R_t}{\pi_{t+1}} \,, \label{eq:fisher}
\end{align}
the real money demand
 \begin{align}
m_t=\chi^{1/\sigma }c_t\big(\frac{R_t-1}{R_t}\big)^{-1/\sigma}\,, \label{eq:md1}
\end{align}
and the labour supply equation based on the real wage $w_t$
 \begin{align}
w_t=c_t^\sigma n_t^\varphi.  \label{eq:ls1}
\end{align}
\subsection{Firms}
A representative firm is assumed with a production function constant in return to scale given by 
%
\begin{align}
y_t=\varepsilon^y_t n_t^{1-\eta} ,
\label{eq:production}
\end{align}
where $\varepsilon^y_t$ is the technology which evolves exogenously with a unit mean. Each period the firm maximizes profits as the difference between production and the wage bill by setting the real wage rate, i.e.
\begin{align}
\max_{w_t}\, y_t - w_t n_t\,,
\end{align}
yielding the optimality condition for wages
 \begin{align}
w_t =(1-\eta) \varepsilon^y_t n_t^{-\eta}  \,. \label{eq:ld}
\end{align}

\subsection{Market Clearing}
We assume that markets clear in every period. The market clearing conditions for the goods market is
\begin{equation}
c_t=y_t\,.\label{eq:gmc}
\end{equation}
The labour market clears combining (\ref{eq:ld}) and (\ref{eq:ls1}), yielding the condition
\begin{equation}
c_t^\sigma n_t^\varphi=(1-\eta)\varepsilon^y_t n_t^{-\eta},
 \label{eq:lmc_opts}
\end{equation}
The market clearing conditions \eqref{eq:gmc}-\eqref{eq:lmc_opts} together with production function \eqref{eq:production} imply that output, consumption and labor depend on technology shock,
\begin{equation}
c_t=y_t=\varepsilon^y_t n^{1-\eta}_t=[(1-\eta){\varepsilon^y_t}^\frac{1+\varphi}{1-\eta} ]^{\frac{1}{\sigma+(\varphi+\eta)/(1-\eta)}}.
 \label{eq:output}
\end{equation}

\subsection{Government Budget Constraint and Policy Rules}
The government issues interest-bearing bonds and non-interesting bearing currency (money), and collects taxes. It operates under the real inter-temporal {\it government budget constraint} (GBC)

\begin{equation}
m_t+b_t+\tau_t=\frac{m_{t-1}}{\pi_t}+R_{t-1} \frac{b_{t-1}}{\pi_{t}}\,, \label{eq:gbc}
\end{equation}
subject to the transversality condition
\begin{equation}
\lim_{j\rightarrow \infty}  \prod_{k=0}^j (\frac{\pi_{t+k}}{R_{t+k-1}}  ) b_{t+j}=0\,.\label{eq:transversality}
\end{equation}

{\it Fiscal policy} takes the linear tax rule as in \cite{leeper1991equilibria}
 \begin{align}
\tau_t=\gamma_0+\gamma b_{t-1}+\varepsilon_t^\tau\,, \label{eq:fp}
\end{align}
where $\varepsilon_t^\tau$ is an exogenous random shock that is assumed to be i.i.d. with mean zero. We also make the natural assumption that $0 \leq \gamma \leq \beta^{-1}$. We follow the terminology of \cite{leeper1991equilibria} to define fiscal policy as being {\it active} if $\gamma <\beta^{-1}-1$ and {\it passive} if $\gamma >\beta^{-1}-1$. 

{\it Monetary policy} follows \cite{BSGU2001b}  and \cite{EH2005} with a global non-linear interest rate rule
\begin{align}
R_t-1=\varepsilon^R_t f(\pi_t). \label{eq:gtaylor}
\end{align}
The function $f(\pi)$ is assumed to be non-negative and nondecreasing, while $\varepsilon^R_t$ is an exogenous, i.i.d. and positive random shock with a mean of one. We adopt the notation $\alpha = f'(\pi_t)$ and use the functional form
\begin{align}
f(\pi_t)=(R^*-1)(\frac{\pi_t}{\pi^*})^{\frac{AR^*}{R^*-1}}\,, \label{eq:f}
\end{align}
where $A>1$, and $\pi^*>1$ is the inflation target of the monetary authority. 
This specification of monetary policy implies that the nominal interest rate is strictly positive and strictly increasing in the inflation rate. We refer to monetary policy as {\it active (passive)} if the monetary authority raises the nominal interest rate by {\it more (less) than one-for-one} in response to an increase in the inflation rate, that is, if $\alpha>(<)1$.

\subsection{Steady State and Rational Expectations Equilibrium}
The deterministic steady states in the absence of random shocks is characterised by the following set of equations:
\begin{align}
&\text{Euler / Fisher Equation:\ \ }  R=\frac{\pi}{\beta} \label{eq:fisher_ss}\\
&\text{Money Demand:\ \ \ \ \ \ \ \ \  \ \ }  m=y\big(\frac{\pi-\beta}{\chi \pi}\big)^{-1/\sigma} \label{eq:md_ss}\\  
&\text{Monetary Policy:\ \ \ \ \ \ \ \ \ \ }   R=1+(R^*-1)(\frac{\pi}{\pi^*})^{\frac{AR^*}{R^*-1}} \label{eq:mp_ss}\\
&\text{Fiscal Policy \& GBC:\ \ \ \ \ }  b=(\frac{1}{\beta}-1-\gamma)^{-1}[\gamma_0+(1-\frac{1}{\pi})m ] \label{eq:fp_ss}\\
&\text{Output:\ \ \ \ \ \ \ \ \ \ \ \ \ \ \ \ \ \ \ \ \ \ \ \  }  y^{\sigma+\frac{\eta+\varphi}{1-\eta}}=1-\eta
\end{align}

Equation \eqref{eq:fisher_ss} and \eqref{eq:mp_ss} together determine the steady state of inflation:
\begin{align}
\frac{\pi}{\beta}=1+(R^*-1)(\frac{\pi}{\pi^*})^{\frac{AR^*}{R^*-1}} \label{eq:infl_ss}
\end{align}

If $f(\cdot)$ is continuous and differentiable as in \eqref{eq:f}, and has a steady state $\pi^*$ with $f'(\pi^*)>1$, in accordance with the Taylor principle given by \eqref{eq:gtaylor}, non-negativity of nominal interest rate implies the existence of a second low inflation steady state $\pi_L$ with $f'(\pi_L)<1$. Figure\ \ref{fig:ss} illustrates this multiplicity of steady-state inflation via the intersection of the Fisher equation and monetary policy. 

\begin{figure}[!ht] \centering%
\includegraphics[height=5.5cm,width=0.45\textwidth]{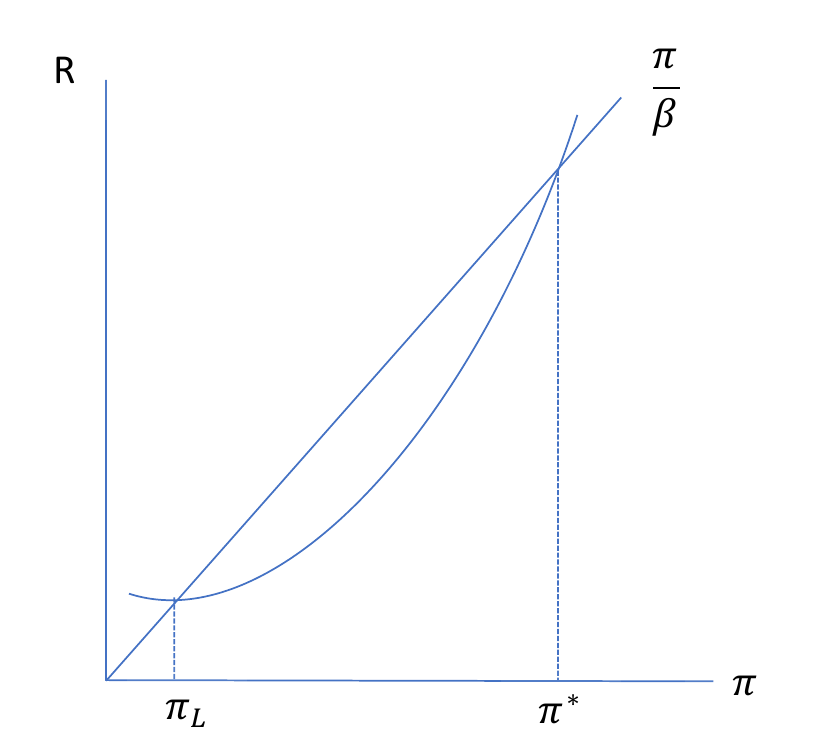} 
\caption{The two steady states of inflation correspond to the intersection between the Fisher equation and the Taylor rule.} \label{fig:ss} 
\end{figure}

These results are formalised by 

\begin{proposition}
\text{[\cite{BSGU2001b}]} 
There exist two steady states of inflation. The first one is an inflation rate $\pi^* \geq 1$ at which the steady state Fisher equation is satisfied and at which the feedback rule is active; that is, $R^*=\frac{1}{\beta}\pi^*$ and $f'(\pi^*)=\frac{A}{\beta}>\beta^{-1}$. The second one is an inflation rate $\pi_L<\pi^*$ such that the steady state Fisher equation is satisfied and the interest rate rule is passive; that is $R_L=\frac{1}{\beta}\pi_L$ and $f'(\pi_L)=\frac{A}{\beta}(\frac{\pi_L}{\pi^*})^{\frac{(A-1)R^*+1}{R^*-1}}<\beta^{-1}$.
\end{proposition}

Once inflation is determined, real money balances are given by \eqref{eq:md_ss}, and real debt is determined by \eqref{eq:fp_ss}. 

In the neighbourhood of either steady state, our model can be described by a linear approximation for $\pi_t$ and $b_t$ of the  form
\begin{align}
\begin{bmatrix}
\hat  \pi_t \\
\hat b_t
\end{bmatrix}=
\mathbf{B}
\begin{bmatrix}
\hat E_t \pi_{t+1} \\
\hat E_t b_{t+1}
\end{bmatrix}+
\mathbf{C}
\begin{bmatrix}
\hat  \varepsilon^R_t  \\
\hat  \varepsilon_{t}^\tau \\
\hat \varepsilon^y_t 
\end{bmatrix}.
\label{eq:linear_mdl}
\end{align}
Note that output is given by \eqref{eq:output}, which is exogenous, depending only on the technology shock. According to \cite{BK1980}, the solution to \eqref{eq:linear_mdl} is locally unique if and only if one eigenvalue is within the unit circle and the other eigenvalue is outside the unit circle. The two eigenvalues of the system \eqref{eq:linear_mdl} are given by $\frac{1}{\alpha\beta}$ and $\frac{1}{1/\beta-\gamma}$ (see the Appendix for derivation) \footnote{The eigenvalues are the inverses of the eigenvalues of the Blanchard-Kahn conditions. This formulation is common in the learning literature, with the expectations operator on the right-hand side.}.

When there is a non-stochastic steady state, it can be shown that stochastic steady states exist their neighbourhoods if the support of the exogenous shocks is sufficiently small. Furthermore, in this case the steady state is locally determinate, provided the corresponding linearisation is determinate. Throughout the paper we assume that the shocks are small in the sense of having small support. Determinacy needs to be assessed separately for the two steady states at $\pi^* $ and $\pi_L$. We have the following formal result:

\begin{proposition}
\text{[\cite{EH2007}]} 
In the linear system given by \eqref{eq:linear_mdl}, \\
(i) If fiscal policy is passive, $|\gamma -\beta^{-1}|<1$, the steady state $\pi^*$ is locally determinate and the steady state $\pi_L$ is locally indeterminate.\\
(ii) If fiscal policy is active, $|\gamma -\beta^{-1}|>1$, the steady state $\pi^*$ is locally explosive and the steady state $\pi_L$ is locally determinate.
\end{proposition}

\begin{proof}
Since $\alpha=f'(\pi)$, it is easy to verify that at the higher steady state $\pi^*$, $|\alpha\beta|>1$ and at the lower steady state $\pi_L$, $|\alpha\beta|<1$.\footnote{More details can be found in \cite{EH2007}, who prove that the linearisation yields a locally unique asymptotically stationary rational expectations equilibrium if monetary policy is (locally) active and fiscal policy is passive, or if monetary policy is (locally) passive and fiscal policy is active.}
\end{proof}

\section{Learning Approaches\label{sec:learning}}
In this section, we first review the adaptive learning approach, one of the main learning approaches used in economics, in light of our model. We then give a general introduction to the main concepts in (deep) reinforcement learning, and to the specific learning algorithm used in this paper. We translate this algorithm to our model, and derive state transition and learning protocols tailored to the model setting. Finally, we put both learning approaches in context using the concept of generalised policy iteration, which offers a unifying framework. 

\subsection{Adaptive Learning\label{sec:adaptive}}
Learning in (macro)economics is a way to deviate from the rational expectation hypothesis in principled ways while still adhering to general equilibrium models \citep{sargent1993bounded,EH2001book,evans2009learning,EP2018imperfect}. In this sense, learning approaches contribute the the study of the general notion of bounded rationality, which may include rational expectations as a special case.

One of the main approaches in the economic learning literature is {\it adapative learning}. Private agents, households in our case, make forecasts using a reduced form econometric model of the relevant variables, and estimate the parameters of this model in a self-referential system based on past data. In each period, the economy uses the agent forecast as input and attains a temporary equilibrium which provides a new data point for the next period's forecast. This sequence of temporary equilibria may generate parameter estimates that converge to a fixed point corresponding to a rational expectation equilibrium for the economy. In this case, the rational expectations equilibrium is stable under learning, or learnable.

\cite{EH2001book} have shown that there is a close connection between the possible convergence of least squares learning to a rational expectation equilibrium and a stability condition, known as {\it E-stability}, based on a mapping from a {\it perceived law of motion} (that private agents are estimating) to an {\it implied actual law of motion} generating the data under these perceptions. E-stability is defined in terms of the local stability at a rational expectations equilibrium of a differential equation based on this map.

If there are multiple rational expectations equilibria, the nature of the perceived law of motion used by the agents in forecasting, i.e. their econometric model, can  determine which equilibria are learnable. This may then serve as a selection process for the study of equilibria under learning. We focus on the case in which the exogenous shocks are i.i.d. processes. The rational expectation solutions of $\pi_t$ and $b_t$ are just noisy steady states, i.i.d. processes. The forecasts of $\pi^e_{t+1}$ and $b^e_{t+1}$ do not depend on the exogenous shocks. It is now natural for private agents to forecast by simply estimating the mean values of $\pi_t$ and $b_t$. It is called {\it steady state learning}. This simplifies our analysis without affecting the theoretical results.

Translating this to our model setting, agents treat \eqref{eq:linear_mdl} as the perceived law of motion, where they estimate the mean of each variable. We can express the expectations of the variables with the estimates of their means. This can be written as a simple recursive algorithms,
\begin{equation}
x^e_{t+1}=x^e_{t}+\phi_t(x_{t-1}- x^e_{t})\,,\label{eq:sslearning}
\end{equation}
with $x\in\{\pi,b\}$. The superscript $x^e_t$ refers to the agent's expected quantity for time $t$. $\phi_t$ is the gain sequence. Under least-squares learning it is usually taken to be $\phi_t=\frac{1}{t}$, often termed a ``decreasing-gain'' sequence, where the influence of new observations decreases over time.\footnote{An alternative is a constant gain, which can have the advantage of reacting better to a changing environment, it will also result more noisy learning and different stability criteria. See \cite{EH2001book}, among others, for details.}

We return to the nonlinear model for the behavioural rules of the agent, such that we can examine the global dynamics of the system. We replace rational expectations with point expectations in the model equations (\ref{eq:fisher}) and (\ref{eq:gbc}), leading to the corresponding nonlinear dynamic system\footnote{The full rational expectations problem can be written as $\mathbb{E}\big[F(y_{t+1},y_t,y_{t-1},\varepsilon_t)\big]=0$ for endogenous state variables $y_t$ and innovations $\varepsilon_t$. Here we make the assumption of point expectations, e.g. replacing $E_t c_{t+1} \pi_{t+1}^\sigma$ with $c^e_{t+1} (\pi^e_{t+1})^\sigma$. For stochastic shocks with small bounded support this is a reasonable approximation.} $F^e$

\begin{eqnarray}
c_t\,&=&\,c^e_{t+1} ({\frac{\pi^e_{t+1}}{\beta R_t}})^\sigma \label{eq:fisher_l}\\
%
%
c_t^{\sigma+\frac{\eta+\varphi}{1-\eta}}\,&=&\,(1-\eta)(\varepsilon^y_t)^{\frac{1+\varphi}{1-\eta}} \label{eq:y_l}\\
\chi^\sigma c^e_{t+1}\big[\frac{ \bar \theta f(\pi^e_{t+1})+1}{\bar \theta f(\pi^e_{t+1})}\big]^{1/\sigma}\,&+&\,b^e_{t+1}+\gamma_0+\gamma b_{t}+\varepsilon_{t+1}^\tau \nonumber \\
\,&=&\,\chi^\sigma \frac{c_{t}}{\pi^e_{t+1}}\big(\frac{R_t-1}{R_t}\big)^{-1/\sigma}+R_{t} \frac{b_{t}}{\pi^e_{t+1}} \label{eq:gbc_l}\\
R_t-1\,&=&\,\varepsilon^R_t f(\pi_t)\label{eq:mp_l}
\end{eqnarray}

The dynamics for $\pi_t$ and $b_t$ under learning is then given by equations \eqref{eq:sslearning}--\eqref{eq:mp_l}. According to \cite{EH2001book}, the local asymptotic stability of the ordinary differential equation
\begin{equation}
\frac{dx^e}{du}=F_x^e(\pi^e,b^e)-x^e\,,
\label{eq:estable}
\end{equation}
again with $x\in\{\pi,b\}$, provides the relevant E-stability criterion for the stochastic model, under steady state learning, when the shocks are small. Here, $u$ denotes notional time, and $F^e(\cdot)$ is the mapping from the perceived law of motion to the corresponding actual law of motion. E-stability is determined by the Jacobian matrix of $F^e(\cdot)$ at the steady state. This is approximated by the matrix $\mathbf{B}$ of \eqref{eq:linear_mdl} evaluated at the steady state. The E-stability conditions are that one eigenvalue of $|\mathbf{B}-I|$ have real part less than zero and the other eigenvalue bigger than zero. The formal learning results are summarised by the following.

\begin{proposition}
\label{prop:3}
Under steady state learning, if the support of shocks are sufficiently small, we have
(i) If fiscal policy is passive, $|\gamma -\beta^{-1}|<1$, the steady state $\pi^*$ is locally stable under learning and the steady state $\pi_L$ is not locally stable under learning.\\
(ii) If fiscal policy is active, $|\gamma -\beta^{-1}|>1$, the steady state $\pi^*$ is not locally stable under learning and the steady state $\pi_L$ is locally stable under learning.
\end{proposition}

\begin{proof}
The eigenvalues of $|\mathbf{B}-I|$ are $ev_1=\frac{1}{\beta f'(\pi)}-1$ and $ev_2=\frac{1}{1/\beta-\gamma}-1$. Since $ f'(\pi^*)>\frac{1}{\beta}$ and $f'(\pi_L)<\frac{1}{\beta} $, therefore we have $ev_1(\pi^*)<0$ and $ev_1(\pi_L)>0$. When fiscal policy is passive, $ev_2>0$ and when fiscal policy is active, $ev_2<0$. 
\end{proof}

\subsection{(Deep) reinforcement Learning}

The problem of maximising the long-run reward of an agent within a modelling environment has been studied extensively in the field of reinforcement learning. The idea of reinforcement learning is to learn behavioural rules, or policies, that, depending on state observations, lead to agent actions that maximise the expected reward. Instead of relying on an analysis of the economic model, which requires extensive knowledge of the model on the part of the agents, reinforcement learning has the promise of finding model solutions with minimal, but flexible, requirements on agents' knowledge. In this section, we give a brief introduction to reinforcement learning and connect it to our model setting from Section \ref{sec:model}. This connection has a number of features that are specific to expectational models used in economics and finance, which can be readily transferred to other model settings. A comprehensive introduction to reinforcement learning is given in \cite{Sutton1998RL}.

\subsubsection{The (deep) reinforcement learning problem}

An agent in a reinforcement learning setting aims to maximise its expected cumulative lifetime reward, or the {\it expected return}, that is
\begin{equation}\label{eq:reward}
\max_{\mathcal{P}} \mathbb{E}_t[G_t]\quad\text{with}\quad G_t\equiv\sum_{k=0}^{\infty} \beta^k r_{t+1+k}(s)\,,
\end{equation}
with $\beta\in (0,1]$ a discount factor and the state-dependent reward $r_t(s) = r(s=s_t)\in \mathbb{R}$ and $s\in\mathcal{S}\subset\mathbb{R}^{n_s}$, $n_s$ being the dimension of the state space. 
The agent achieves maximisation of (\ref{eq:reward}) by optimising its behavioural rules, or policy, $\mathcal{P}: s_t \rightarrow a_t\in\mathcal{A}\subset\mathbb{R}^{n_a}$ ($n_a$ being the dimension of the action space) based on observed state transitions.\footnote{The policy is usually denoted $\pi$ in the reinforcement learning literature, with $\pi^*$ denoting the optimal policy. However, these two expressions are reserved for inflation and target inflation in macroeconomics, such that we denote agent (optimal) policies by $\mathcal{P}^{(*)}$.}. These actions interact with the environment the agent is living in leading to the next state and returning a reward, i.e. $\mathcal{E}: (s_t, a_t) \rightarrow (s_{t+1},r_{t})$. This process is schematically shown in Figure \ref{fig:rl}. At each time step $t$, the agent observes the state $s_t$, takes action $a_t$, while the environment returns reward $r_{t+1}$ and a new state  $s_{t+1}$ to the agent, which acts again, and so on.
\begin{figure}[tbp] \centering%
\includegraphics[width=0.8\textwidth]{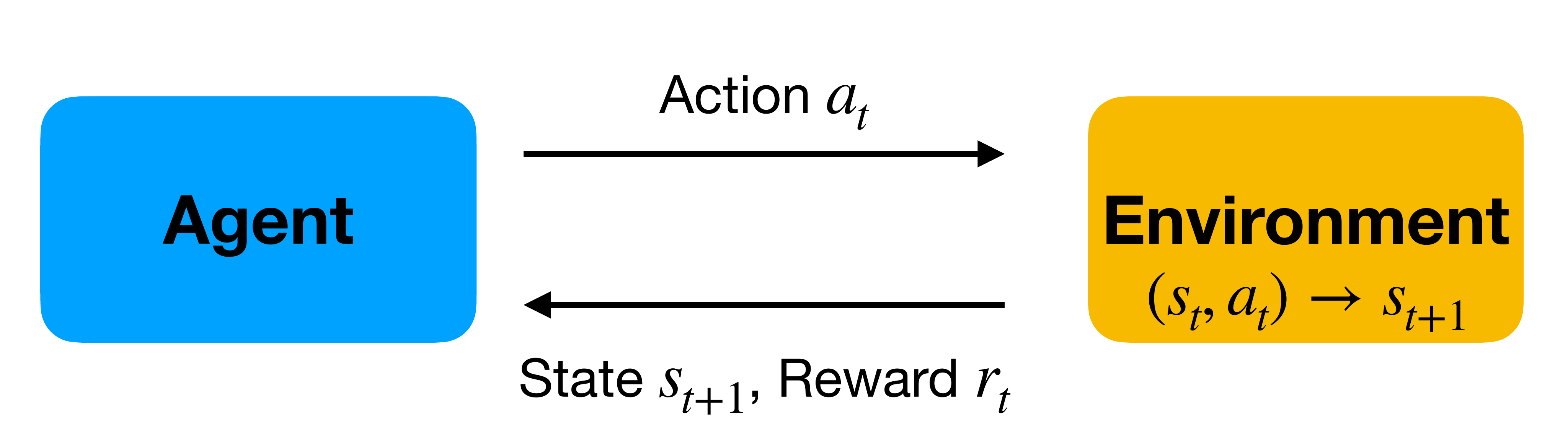} 
\caption{Agent-environment interaction in reinforcement learning.} \label{fig:rl} 
\end{figure}%
This can be formulated as a Markov decision process defined by the tupel ($\mathcal{S}$, $\mathcal{A}$, $\mathcal{T}$, $r$). The transition probability $\mathcal{T}:\mathcal{S}\,\times\,\mathcal{A}\times\,\mathcal{S}\rightarrow [0,1]$ describes the probability of the next state $Pr(s_{t+1}|s_t,a_t)=F(s_{t+1}|s_t,a_t)$ given the current state $s_t$ and action $a_t$, where $F(\cdot)$ describes the model environment. This transition function fulfils the Markov property, i.e. it only depends on the current state and action, but not the history of state transitions.

Finding the optimal policy $\mathcal{P}^*$ can be approached from the {\it state-value function} following a policy $\mathcal{P}$
\begin{eqnarray}
V_{\mathcal{P}}(s)\,&=&\,\mathbb{E}_{\mathcal{P}}\big[G_t|s=s_t\big] \label{eq:state_values}\\
\,&=&\, \max_{a\in\mathcal{A}} \mathbb{E}_{\mathcal{P}}\big[ G_t | s=s_t,a=a_t\big] \nonumber\\
\,&=&\, \max_{a\in\mathcal{A}} Q(s,a)\,, \label{eq:q_value}
\end{eqnarray} 
where the last expression defines the {\it action-value function}, i.e. the expected return following a behavioural rule $\mathcal{P}$ given a state and action. The optimal policy $\mathcal{P}^*$ maximises both state and state-actions values, which also maximises expected return - our final goal.

The state-action value function fulfils the recursive  {\it Bellman equation}
\begin{equation}
\label{eq:bellman}
Q(s_t,a_t) = r(s_t) + \beta \mathbb{E}_{\mathcal{P}} \big[ Q(s_{t+1},a_{t+1})\big].
\end{equation}
The current state-action value is the current reward plus the expected value of the next state. These components form the backbone of so-called actor-critic approaches, where one aims to improve the actor $a_t = \mathcal{P}(s_t)$ which is evaluated by the critic $Q(s_t,a_t)$. That is, we will be estimating two separate quantities, $\mathcal{P}$ and $Q$, where the one can be used to evaluate the other. These can be parameterised using general function approximators in the form of artificial neural networks with internal weights $\phi$ and $\theta$, denoted by $\mathcal{P}_{\phi}$ and $Q_{\theta}$, respectively. The use of deep artificial neural networks \citep{goodfellow2016deep} is more generally at the forefront of advances in solving complex tasks like computer vision, speech recognition or robotic navigation. {\it Deep reinforcement learning}, as used in this analysis, combines the traditional reinforcement learning approach to dynamic problems and the use of deep artificial neural networks. The general function approximator properties of the latter dramatically increase the capabilities of the former making possible recent advances, including the analysis in our study.  
%

Now, using sampled state transitions as observations, i.e. interactions of the agent and the environment, and standard optimisation techniques like stochastic gradient descent, the policy and action-value function networks can be trained by iteratively minimising the Bellman residuum, 
\begin{eqnarray}
L(\phi,\theta)&\,=\,&\mathbb{E}_{s_t,a_t,r_t}\bigg[\frac{1}{2}\big(Q_{\theta}(s_t,a_t)- \hat{Q}_{\theta}(s_t, a_t)  \big)^2\bigg] \label{eq:bell_resid}\,,\\
\text{with}\quad \hat{Q}_{\theta}(s_t, a_t) &\,=\,& r_t(a_t,s_t) + \beta\, \mathbb{E}_{\mathcal{P}}\big[ Q_{\theta}\big(s_{t+1},\mathcal{P}_{\phi}(s_{t+1})\big)\big]\,.
\end{eqnarray}

The details of how to calculate the needed gradients vary by algorithm, while the development of efficient and stable approaches to solve the above problem in different settings is the subject of ongoing research in artificial intelligence. In this paper we  use the soft actor-critic approach of \cite{Haarnoja2018sac}.\footnote{Maximum entropy approaches like this have the advantage that they produces relatively stable learning outcomes compared to, for example policy-gradient approaches like \cite{lillicrap2019continuous}. Note, however, that we are less concerned in the precise learning algorithm used rather then in its ability to solve the agent's optimisation problem. The code we used for optimisation is available at \url{https://github.com/pranz24/pytorch-soft-actor-critic}.} We have now defined the general setting in deep reinforcement learning. We next relate this setting to our model environment.

\subsubsection{Deep reinforcement learning in the context of expectational macroeconomics}

The household's problem (\ref{eq:value}) is analogous in structure to the learning agent's problem (\ref{eq:reward}) when replacing the general reward with the household's utility (\ref{eq:utility}).

The environment $\mathcal{E}$, about which the agent is ignorant, is given by the production process (\ref{eq:production}), goods pricing (\ref{eq:ld}),\footnote{The prices of goods are assumed to be set optimally by the firm. This could be relaxed by extending the current setting to a multi-agent problem, where the firm would have to learn about its pricing strategy. Multi-agent learning problems are considerably more complex than single-agent problems and we leave this setting for future work.} the government budget constraint (\ref{eq:gbc}), market clearing (\ref{eq:gmc}), fiscal policy (\ref{eq:fp}) and monetary policy (\ref{eq:gtaylor}). It does not include the optimality conditions \eqref{eq:fisher}--\eqref{eq:ls1}, which the household has to learn about using deep reinforcement learning.

We can make a connection here with dynamic programming and value iteration methods which  have been used in the economics literature. The difference to the current setting is that these approaches assume that the agent possesses knowledge about the model probabilities $F(s_{t+1}|s_t,a_t)$. In deep reinforcement learning this knowledge is mostly absent or intractable. It is part of the agent's learning task is to infer the probabilities from observations that are generated through interactions with the environment.

Now, the state at time $t$, $s_t$, is given by last period's money, bond holdings, inflation, consumption, and hours worked, as well as the exogenous components of this period's fiscal policy, monetary policy, and technology shocks:
\begin{equation}
\label{eq:state}
s_t\,=\,\left( m_{t-1}, b_{t-1}, \pi_{t-1}, c_{t-1}, n_{t-1}, \epsilon_{t}^{\tau},\epsilon_{t}^{R}, \epsilon_{t}^{y} \right)\,.
\end{equation}
The  state representation is not unique,\footnote{Inflation can be replaced by the gross interest rate set by the monetary authority according to (\ref{eq:gtaylor}).} but it does need to fulfil the Markov property of allowing for state transitions only based on the knowledge of the current state and not past states.

 The household's actions $a_t$ at each time step $t$ are a tuple of consumption, bond saving and hours worked, denoted by
\begin{equation}
\label{eq:actions}
a_t\,=\,\left( c_t^{act}, b_t^{act}, n_t \right)\,, 
\end{equation}
where $x^{act}_t$, $x\in\{c,b\}$, represents actions with reference to last period's price level, i.e. $X_t/P_{t-1}$ with $X_t$ being nominal consumption or bond holdings. The state variables in \eqref{eq:state} are now determined by the interactions of the household's actions \eqref{eq:actions} and the model environment. These actions set the level of inflation, real consumption and real bond holdings according to this period's prices, i.e. 
\begin{eqnarray}
\pi_t &\,=\,& c_t^{act}/y_t\,,\label{eq:rl_inflation}\\
c_t   &\,=\,& c_t^{act}/\pi_t\,, \label{eq:rl_cons}\\
b_t   &\,=\,& b_t^{act}/\pi_t\,.\label{eq:rl_bonds}
\end{eqnarray}

The first relation stands for prices clearing markets (\ref{eq:gmc}) where the relationship between the agent's actions (here, choosing consumption with reference to last period's price level) and the price adjustment process between different periods has been made explicit. This mechanism respects the information flow in the model, with price inflation relating real quantities between periods. The agent observes the state with reference to past period's price level $P_{t-1}$ and takes its actions accordingly. This period's price level $P_{t}$ is then set via the market clearing condition (\ref{eq:rl_inflation}), determining this period's inflation $\pi_t$. This mechanism incorporates the dynamics of non-stationary nominal quantities within a stationary real setting, which is particular to  economics settings. That considerably simplifies the learning process, because learning a non-stationary environments is considerably more challenging.

Learning happens in episodes. Each episode is initiated with a random state drawn uniformly from a region of interest in the state space. Subsequently, iterations between agent actions and the environment result in state transitions. The agent's parameters, $Q_{\theta}$ and $\mathcal{P}_{\phi}$, are updated between such steps. An episode ends when a termination criterion is reached. This is in our case, either a maximal number of steps $N_{epi}^{max}$ or an improvement in agent utility below a fixed small threshold $d_u^{min}$. The rationale for using an episodic termination criterion is that this allows the agent to experience more regions of the state space during learning, and this this is coupled to learning progress in the case of $d_u^{min}$.\footnote{The learning problem can alternatively be formulated without the termination criterion on $d_u^{min}$ as the learning task is open-ended.}

Training/learning: Parameter updates and explorative actions characterise the training or learning of the agent. Explorative actions are actions which are not optimal according to the currently learned behavioural policy, but have noise components to them. This is a crucial part in deep reinforcement learning, as it allows the agent to discover new and ultimately better actions. The magnitude of the random component in actions characterises the exploration-exploitation trade-off in deep reinforcement learning. Set too small or large, the agent will fail to effectively learn.

Our action space is continuous, and exploration is achieved by drawing from a normal distribution generated from $\mathcal{P}_{\phi}(s_t)$.\footnote{We follow \cite{Haarnoja2018sac} which take the action $\mathcal{P}_{\phi}(s_t)$ for the mean and logarithm of the standard deviation of that normal distribution projected onto the maximally allowable action space.} That is, most actions will be close to the current best action, the mean of this action distribution, while deviations from this mean explore the action space. If such actions turn out to be beneficial, i.e. return higher utility to the household given the current state, the action network  $\mathcal{P}_{\phi}$ will move into this direction during learning.
 
For updates of the parameter in $(\mathcal{P}_{\phi},Q_{\theta})$, the agent draws randomly from a fixed-size memory of experience consisting of $N_{mem}$ past state transitions, e.g. to perform stochastic gradient descent.\footnote{The oldest transition drops out if the memory is full.} The overall training phase is set to last for a maximal number of steps $N_{train}$, i.e. the number of parameter updates by which we expect the agent to solve its optimisation problem. In the current setting, we consider the agent's problems as solved if the household finds action values corresponding to one of the steady states of our model, and if episodes terminate in such a state.

Testing: The evaluation of learning goals, such as the distance to a steady state in the state space, happens during testing. Testing consists of a set number of test episodes $N_{test}$ that the agent runs through between a fixed number of training steps $N_{interval}\ll N_{train}$.\footnote{This is likely to happen within an unfinished training episode, which is paused at this point to be resumed after testing. The motivation to use a fixed number of training steps between test episodes instead of a fixed number of training episodes is that the length of an episodes during training or testing is stochastic and may also change during learning, such that a number of steps allows us to measure the agent's learning progress uniformly.}

 The full training-cum-test setting is summarised in Algorithm\ \ref{algo:train_test}. The initial $N_{burn}$ steps of pure random actions serve as a burn-in phase to give the agent some orientation before exploration starts. We save all test transitions, as well as the agent parameters $(\mathcal{P}_{\phi},Q_{\theta})$ at different stages of learning. This allows ex-post experimentation, the reproduction of test results, or the the flexible adjustment of the learning setting.
 
 \begin{algorithm}
    \caption{Training and testing protocol of household agent}\label{algo:train_test}
    Initialise: Environment $\mathcal{E}$ (parameterised model), agent (parameterised by $\mathcal{P}_{\phi}$, $Q_{\theta}$)
    \begin{algorithmic}
        \FOR{steps = $1$ to $N_{train}$}
        	\STATE initialise training episode with random state $s_t$
            \WHILE{training episode is not done}
            	\IF{$\text{steps}\leq N_{burn}$}
            		\STATE Take allowed random action $a_t$
            	\ELSE
                	\STATE Draw exploration action $a_t=\mathcal{P}_{\phi}^{exp}(s_t)$
                \ENDIF
                \STATE Environment returns $(r_{t}, s_{t+1})=\mathcal{E}(s_t, a_t)$
                \STATE Add transition $(s_t,a_t,r_{t},s_{t+1})$ to memory
                \STATE Update $\mathcal{P}_{\phi}$, $Q_{\theta}$ using batch gradient descent from memory
                \IF{$mod(\text{steps},N_{interval})\,=\,0$}
                	\FOR{test episode = $1$ to $N_{test}$}
                		\STATE Record state transitions (*)
                	\ENDFOR
                	\STATE Save current agent ($\mathcal{P}_{\phi}^{steps}$, $Q_{\theta}^{steps}$)
                \ENDIF
                \STATE State update $s_t \leftarrow s_{t+1}$
                \STATE Test episode termination criteria ($N_{epi}^{max}$, $d_u^{min}$)
            \ENDWHILE
        \ENDFOR
    \STATE Save final agent ($\mathcal{P}_{\phi}^{final}$, $Q_{\theta}^{final}$)
    \end{algorithmic}
\end{algorithm}

 We still need to define the state transition of a single (testing) step (* in Algorithm \ref{algo:train_test}) which also includes the economics of the household's learning problem. 

{\it Step sequence (*) for single transition:} $s_t\rightarrow s_{t+1}$
\begin{enumerate}
\item Observe state $s_t$
\item Take actions $\mathcal{P}_{\phi}(s_t)=a_t=\left( b_t^{act}, c_t^{act}, n_t \right)$ [$\mathcal{P}_{\phi}^{exp}$ is used during training for exploration]
\item Production $y_t$ takes place according to (\ref{eq:production}) and firm sets wages using (\ref{eq:ld})
\item Markets clear: Inflation $\pi_t$ is set by (\ref{eq:rl_inflation}) 
\item This determines real consumption $c_t$ and real bond holdings $b_t$ according to (\ref{eq:rl_cons})-(\ref{eq:rl_bonds})
\item Policy realisations:
	\begin{itemize}
	\item The monetary authority sets the current gross interest rate $R_t$ based on $\pi_t$ via the Taylor rule (\ref{eq:gtaylor})
	\item The government raises taxes $\tau_t$ (\ref{eq:fp})
	\end{itemize}
\item The money holdings $m_t$ are realised from the GBC (\ref{eq:gbc})
\item Agent obtains reward $r_t=U(c_t,m_t,n_t)$
\item Next periods shocks are realised, $(\epsilon_{t+1}^{\tau},\epsilon_{t+1}^{R}, \epsilon_{t+1}^{y})$
\item State update $s_t\leftarrow s_{t+1} = \left( m_{t}, b_{t}, \pi_{t}, c_t, n_t, \epsilon_{t+1}^{\tau},\epsilon_{t+1}^{R}, \epsilon_{t+1}^{y} \right)$
\end{enumerate}

Looking at the learning algorithm and state transition, the state representation (\ref{eq:state}) can be roughly partitioned into three groups of variables, those used for state transitions $(b_{t-1}, \pi_{t-1})$, the evaluation of convergence $(c_{t-1}, n_{t-1}, m_{t-1})$, and observed shocks $(\epsilon_{t}^{\tau},\epsilon_{t}^{R}, \epsilon_{t}^{y})$. The first and the last group are needed for non-terminal state transitions. The second group, together with last period's money holdings $m_{t-1}$, is used to evaluate the termination criterion if an episode will be terminated and a new one initiated if $d_u\equiv|U_t-U_{t-1}|<d^{min}_u$.\footnote{The termination criterion based on the maximal number of steps within an episode $N_{epi}^{max}$ is independent of the state.} The higher the threshold $d^{min}_u$ is, the more episodes the agent will go through during training, thus exploring more of the state space, with a smaller chance of getting stuck during learning. However, too high a value will lead to imprecision as episodes terminate before the agent reached sufficient precision in its actions potentially ending up far from a steady state value.

\subsubsection{Generalised Policy Iteration}
\label{sec:gpi}

Our two learning approaches, adaptive learning and deep reinforcement learning, can be conceptually compared from the point of view of {\it generalised policy iteration} (GPI). In GPI, policy evaluation, which delivers the value of a state given a policy, and policy improvement, which delivers the change of behaviour to attain higher returns, interact iteratively \citep{Sutton1998RL}. This is depicted in Figure\ \ref{fig:gpi}. Under the (non-trivial) assumption that learning converges, this process results in a fixed point of optimal policy $\mathcal{P}^{*}(s)$ and maximal-return state values $V_{\mathcal{P}^{*}}(s)$. Both together specify the rational expectations equilibrium in our model. As long as the agent has not converged to this point, it is called {\it boundedly rational}.

\begin{figure}[tbp] \centering%
\includegraphics[width=0.7\textwidth]{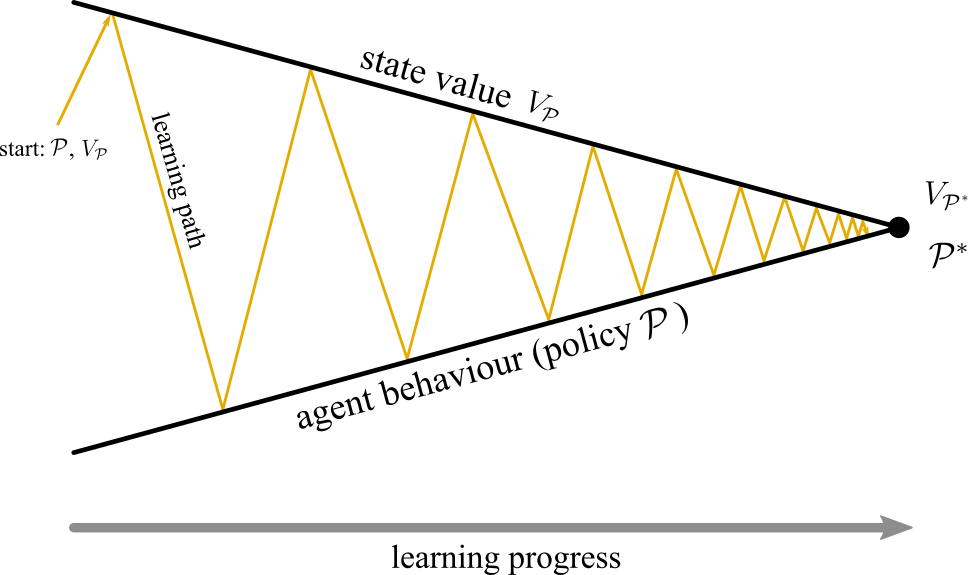} 
\caption{Schematic representation of generalised policy iteration.} 
\label{fig:gpi} 
\end{figure}%

In our modelling environment, optimal behaviour is given by the household first-order conditions (FOC) (\ref{eq:euler})--(\ref{eq:ls}), while (locally) optimal state values are the two steady states corresponding to $\pi^*$ and $\pi_L$.  Adaptive learning fixes the agent's behaviour by using these FOC within the system equations (\ref{eq:fisher_l})--(\ref{eq:gbc_l}). This can be thought of as a horizontal policy line in Figure\ \ref{fig:gpi}. The state equations (\ref{eq:sslearning}) describe state values, and converge according to Proposition\ \ref{prop:3}. The difference between adaptive learning and deep reinforcement learning is that the latter does not fix behaviour, but both the agent policy $\mathcal{P}$ and the corresponding state-action values $Q$ are learned simultaneously. The GPI framework will be useful to understand and quantify learning behaviour.\footnote{GPI can also be used describe {\it Euler learning} in macroeconomics where the behavioural rules differ from the FOC (see e.g. \cite{EP2018imperfect} for a brief discussion). In a nutshell, it may be less desirable from the GPI perspective as convergence cannot happen if the behavioural rules are not flexible enough to converge to the FOC.}

\section{Learning results}
\label{sec:results}

We analyse household learning within the deep reinforcement learning framework and the model presented in Section (\ref{sec:model}). The main parameterisation used in our analysis is given in Table\ \ref{tab:para}. The motivation for these choices is a trade-off between clarity of presentation and realism within our arguably simple model.

Table \ref{tab:para} reports the calibration. We take the model's frequency to be quarterly\footnote{The process of informational updates between a perceived old price level, actions and adjustment (\ref{eq:actions})-(\ref{eq:rl_bonds}) may be more realistic on a quarterly frequency compared to an annual one. On the other hand, public inattention to or unawareness of macroeconomic aggregates may also support an annual approach \citep{SIMS2010inattention}}and set $\beta=0.99$ which implies a steady state real interest rate of about 4 percent; $\varphi=1$, implying a unitary Frisch elasticity of labor supply; $\sigma=3$, which is within the range of 1 to 3.5 in the literature; $\chi=0.1$, following \cite{EH2005}. The Taylor rule coefficient $A =1.3$ gives two steady states of inflation,  one is $\pi^*=1.01$ (4\% net per annum), the other one is $\pi_L=1.0014$ (liquidity trap). 

\begin{table}[!ht]
\centering
\begin{tabular}{cll}
\toprule
parameter & value & description \\
\hline
$\beta$   			& 0.9900  & discount factor \\
$\sigma$ 			& 3.0000  & inverse of intertemporal elasticity of consumption and money holdings \\
$\eta$    			& 0.001  & production scaling factor \\
$\varphi$    		& 1.0000  & inverse of Frisch elasticity of labor supply\\
$\chi$     			& 0.1000  & relative preference weight of money holdings \\
$\gamma_P$ 			& 0.0200  & passive fiscal policy (PFP) coefficient \\
$\gamma_A$ 			& 0.0000  & active fiscal policy (AFP)  coefficient \\
$A$		  			& 1.3000  & Taylor rule coefficient\\
$\pi^*$	  			& 1.0100  & target gross high-inflation rate (4\% net per annum)\\
$\pi_L$	  			& 1.0014  & implied gross low-inflation steady state (see Figure\ \ref{fig:ss})\\
$\epsilon^{\tau}_t$ & 0.0005  & monetary policy shock (std. dev.)\\
$\epsilon^{R}_t$    & 0.0005  & fiscal policy shock (std. dev.)\\
$\epsilon^{y}_t$	& 0.0005  & technology shock (std. dev.)\\
\bottomrule
  \end{tabular}
\caption{Baseline model parameterisation. The shock series $\epsilon^{\tau}_t$, $\epsilon^{R}_t$, $\epsilon^{y}_t$  follow log-normal, normal and normal distributions, with means of one, zero and one, respectively.}
\label{tab:para}
\end{table}

Steady state values for high and low inflation, as well as for passive and active policy, are given in Table \ref{table:cali_ss}. For better comparability of regimes, the fiscal policy intercept $\gamma_0$ is calibrated such that bond holdings equal annualised output for each policy combination. This does not affect the local stability properties of the model or the learning dynamics of the agent. Steady state money holdings are between 40-50\% of output which is about double the amount of the long-term average of narrow money holdings (M1 in the US) and half the amount of broad money (M3 in the US). While strictly speaking our model only talks to narrow money, extension with a financial sector would represent broader aggregates. Annualised net target inflation of $4\%$ is well above the mandate or recent experience in most advanced economies but well within those of emerging markets. Note that money holdings and household utility are generally higher in the low-inflation steady state $\pi_L$ with passive monetary policy.

\begin{table}[!ht]
\centering
\begin{tabular}{c|cccc}
\toprule
 & \multicolumn{2}{c}{AMP} & \multicolumn{2}{c}{PMP}\\ 
               				 & PFP & AFP & PFP & AFP\\ 
\midrule
$\pi_{ss}$	           		 &  1.0100  &  1.0100  &  1.0014  &  1.0014 \\
$m_{ss}$           			 &  1.7157  &  1.7157  &  2.0614  &  2.0614 \\
$c_{ss}$/$n_{ss}$/$y_{ss}$   &  1       &  1       &  1       &  1      \\
$b_{ss}$           			 &  4   	&  4       &  4       &  4      \\
$u_{ss}$           			 & -1.0170	& -1.0170  & -1.0118  & -1.0118 \\ 
$\gamma_0$   				 & -0.0566  &  0.0234  & -0.0426  &  0.0375 \\            
\bottomrule
  \end{tabular}
\caption{Steady state values under different policy regimes: active/passive monetary policy (AMP/PMP, $\pi^*$/$\pi_L$) and passive/active fiscal policy (PFP/AFP). Source: Authors' calculations.}
\label{table:cali_ss}
\end{table}

\subsection{Learnability of steady states}
\label{sec:learnability}

\subsubsection{Adaptive learning}
Learnability under adapative learning for the different policy regimes in Table\ \ref{table:cali_ss} is determined via the eigenvalues of matrix $\mathbf{B}$ in (\ref{eq:linear_mdl}) and described in Proposition 2. This is graphically summarised in Figure\ \ref{fig:determinacy7} for the parameterisation in Table\ \ref{tab:para}. The horizontal axis shows the fiscal response parameter $\gamma$ and the vertical axis inflation. For monetary policy, the two steady steady inflation values correspond to the intersections in Figure\ \ref{fig:ss}, which are indicated by the horizontally dashed lines. They mark the active and passive monetary policy regimes at $\pi^*$ (AMP) and $\pi_L$ (PMP), respectively. For fiscal policy, the centre vertical grey line marks the boundary between passive and active fiscal policy. Passive and active fiscal policy as used in our analysis, are marked by the vertical red lines at $\gamma_P=0.02$ (PFP) and $\gamma_A=0$ (AFP), respectively. The cross of solid grey lines separate the four policy regimes. In the adaptive learning literature the learnability criterion from Proposition 2 is used as a selection criterion for policy regimes. The determinate regimes (AMP-PFP and PMP-AFP) are learnable while the explosive (AMP-AFP) and indeterminate (PMP-PFP) regimes are not. We now investigate which of these policy regimes are learnable by deep reinforcement learning.

\begin{figure}[!ht] \centering%
\includegraphics[width=0.55\textwidth]{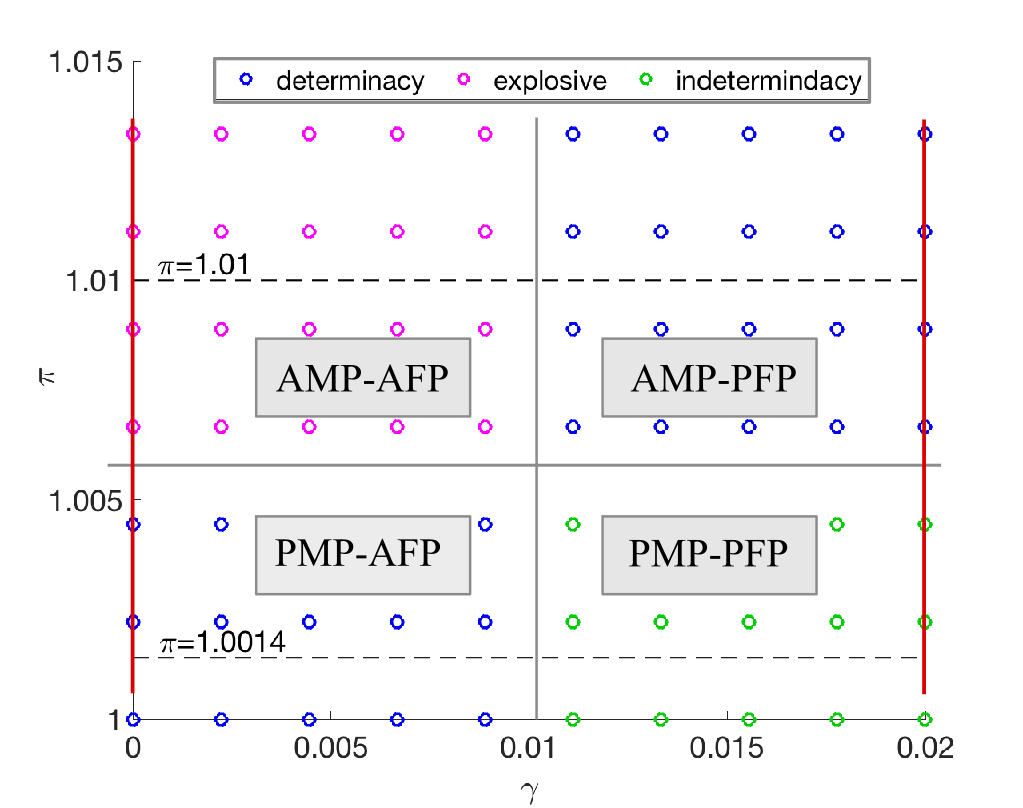} 
\caption{Local dynamic stability properties of steady state values by policy regime. Source: Athors' calculations.} \label{fig:determinacy7} 
\end{figure}

\subsubsection{Deep reinforcement learning}
While there exists a correspondence between the dynamic stability of non-linear maps and local linear approximations, there is no such general relation for deep reinforcement learning problems. Hence, learnability is an empirical question, which we address via numerical simulations. In each of our experiments, we define a region of interest in the household's action space and the state space around either the low or high inflation steady state according to the policy regime of interest. We follow the learning protocol described in Algorithm \ref{algo:train_test}. Details of the actions, states and settings of the learning algorithm following \citep{Haarnoja2018sac} are listed in Table\ \ref{tab:para_2} in the Appendix. Exogenous shocks are disabled in this and the following analyses. Their effects are investigated separately in Section\ \ref{sec:uncertainty}. Most experiments run for a total of $N_{train}=2.5e6$ training steps. We conduct $N_{test}=10$ test episodes for each ten thousand learning steps, i.e. $N_{interval}=1e4$. We say that a state or regime is learnable if the household's action values converge to the corresponding steady state values. The sensitivity of our results to some of the choices in the learning setting is discussed below.\\[.2cm]

\begin{table}[!ht]
\centering
\begin{tabular}{c|cccc}
\toprule
 & \multicolumn{2}{c}{AMP} & \multicolumn{2}{c}{PMP}\\ 
               & PFP & AFP & PFP & AFP\\ 
\midrule
AL       & yes  & no & no & yes \\
DRL	     & yes & yes & yes$^{\dagger}$ & yes$^{\dagger}$ \\
\hline
 & \multicolumn{4}{c}{$\Delta_{ss}$ (\%) for DRL}\\
\hline
$\pi$    & 0.067  & -0.159  & 9.199  & 5.209  \\
$b$      & 0.001  & 0.002  & -0.038  & -0.024  \\
$n$      & 0.000  & 0.002  & -0.007  & -0.001  \\
$m$      & -0.022 & 0.048  & -11.536 & -7.174  \\
$u$      & 0.001  & -0.002  & 0.345 & 0.192  \\
\hline
 & \multicolumn{4}{c}{$|\Delta_{ss}|$ (\%) for DRL}\\
\hline
$\pi$  & 0.346  & 0.278  & 9.217  & 5.209  \\
$b$    & 0.005  & 0.004  & 0.038  & 0.024  \\
$n$    & 0.004  & 0.003  & 0.009  & 0.003  \\
$m$    & 0.091  & 0.089  & 11.569 & 7.364  \\
$u$    & 0.003  & 0.003  & 0.346 & 0.196 \\
\bottomrule
  \end{tabular}
\caption{Comparison of learnability of different policy regimes for adaptive learning (AL) and deep reinforcement learning (DRL): active/passive monetary policy (AMP/PMP at $\pi^*$/$\pi_L$) and passive/active fiscal policy (PFP/AFP). $\Delta_{ss}$ and $|\Delta_{ss}|$ measure the mean and mean absolute difference in percentage from their respective steady state values for the end of test episodes during the last $5e5$ steps of DRL training. The numbers of inflation refer to difference to net inflation, i.e. $1$ and $0.14$ for $\pi^*$ and $\pi_L$, respectively. $^{\dagger}$ indicates that action values are mostly learned well, but with some discrepancy in the learning of steady state money holdings. Source: Authors' calculations.}
\label{table:ss_conv}
\end{table}

{\bf AMP-PFP regime:} This is the classically considered policy regime of monetary dominance around the target inflation level $\pi^*$.  All results we present in the following are taken from test episodes between learning intervals. We focus on the end of these episodes, that is, the final action and state values of each test episode to assess the state of convergence of the household's behaviour and how it compares with steady state values.

The convergence of household actions during learning to their respective steady state values is shown in Figure\ \ref{fig:ss_AMP_PFP}. Consumption choices are implied by inflation via the market clearing condition \eqref{eq:gmc}. The vertical axis shows the distance to steady state values relative to the maximal test distance observed. This normalisation allows for the uniform comparison of learning behaviour for different actions as their numerical scales differ. All lines are moving averages over 25 learning intervals and test cycles. To account for volatility in learning outcomes, e.g. due to the randomisation of episodes' initial states and differences in convergence, 95\% confidence intervals are indicated by the shaded areas.\footnote{We take two standard deviations of the rolling average in both directions.} Despite having no shocks in the model, the finite confidence intervals come from the random initialisation of test episodes which lead to slightly different end points. These are distributed around the steady state values of the model in the rational phase.

We see that households learn the optimal steady state solution well, with all actions converging synchronously, and that actions subsequently stay at this point. The initial increase in the learning distance is due to the `breaking of randomness' at the beginning of learning. This randomness comes form the random initiations of neural network weights. Initial random behaviour can on average be close to steady state values, because experiences are sampled uniformly around them, such that errors cancel each other out. However, we expect a large variance in outcomes during this initial random phase, which is indeed the case as can be seen from the wide confidence intervals at the beginning of learning. In summary, we can divide  the learning process into three phases, an initial random phase, an intermittent learning phase and a terminal rational phase. This final phase means that the household has learned the rational expectation solution, which we will quantify in more detail below.

We report learning behaviour on this normalised scale, because absolute deviations from steady state values depend on several factors of both the learning algorithm and model parameterisation. However, their magnitudes may nevertheless be instructive, as we see the maximal distance an agents behaviour deviates from the optimal during learning conditioned on the settings of the experiment. The peak distances at the transition to the learning phase in Figure\ \ref{fig:ss_AMP_PFP} for net inflation, bond holdings and hours worked are (11.07\%/11.07\%), (0.60\%/29.92\%) and (0.14\%/7.12\%), respective. The first number refers the distance to the steady state value and the second compares this distance to corresponding action range in Table\ \ref{tab:para_2}. We will look at the interpretation of these number in more detail when discussion inflation expectations in more detail below.

\begin{figure}[!ht] \centering%
\includegraphics[width=0.7\textwidth]{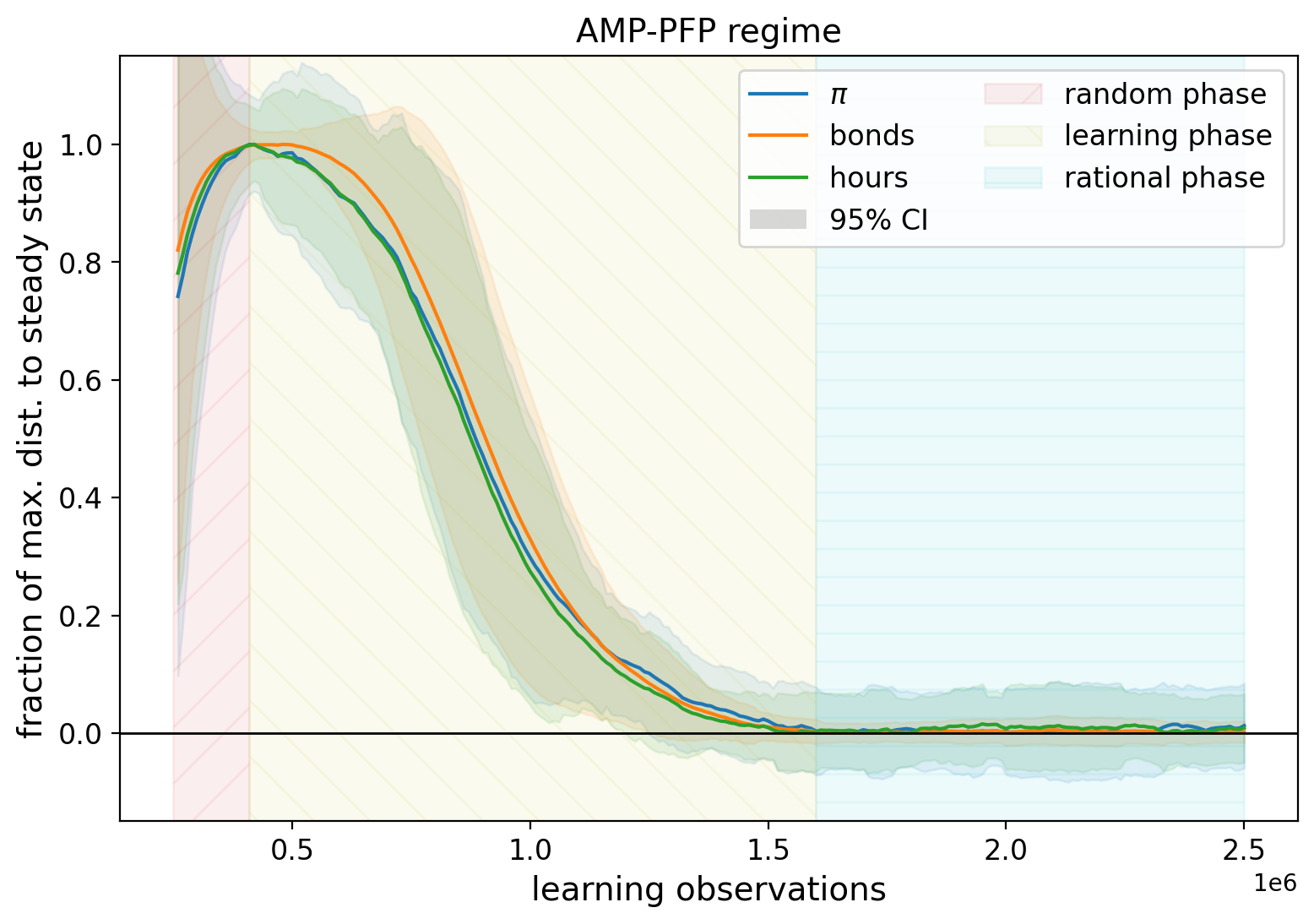} 
\caption{Steady state test convergence of housefold actions for AMP-PFP regime at end of episodes. Shaded bands show 95\% confidence intervals. Hatched areas show learning phases. Source: Authors' calculations.} \label{fig:ss_AMP_PFP} 
\end{figure}%

{\bf All policy regimes:} The learning results are summarised in Figure\ \ref{fig:ss_conv} and Table\ \ref{table:ss_conv} (the learning of money holdings is shown separately in Figure\ \ref{fig:ss_money} in the Appendix). The learning behaviour is qualitatively the same for all policy regimes. We conclude that all regimes are learnable by deep reinforcement learning unlike to adaptive learning. This means that the local linear dynamic properties of the model around a steady state are not necessarily a selection criterion if this state can be attained by the agent. The reason for this is that the household under deep reinforcement learning is not bound by the dynamics of the linearised system, but instead by the ``global map''  given by long-term utility maximisation. In the linear system wandering off the optimal path may lead to unstable (learning) dynamics, while in deep reinforcement learning such an action may just be identified as an action improvement rewarded by higher utility. Economically, this means that all policy regimes can become entrenched if the household spends enough time close to the corresponding steady state. More generally, this shows that deep reinforcement learning is a ``global''  solution technique, as it is capable of retrieving multiple steady state solutions without the need to specify a localised approximation.

However, there are also differences in deep reinforcement learning outcomes between the different policy regimes. These are again not related to the dynamic properties of the linearised system but rather to the state of monetary policy. The regimes of passive monetary policy, i.e. where the slope of the Taylor rule (\ref{eq:gtaylor}) is smaller than $\pi/\beta$, requires more time to be learned with less precision. This is especially the case with passive fiscal policy, where convergence takes considerably more observations, and is more noisy.

The details are instructive from an economic and methodological perspective. Table\ \ref{table:ss_conv} shows that the imprecision in learning is related to the monetary policy regime. Most of the deviation from steady state values stems from inflation, and to a lesser degree from bond holdings.\footnote{Note that the inflation values relate to net inflation at $\pi_L$ which is 0.014.} This translates into larger average deviations of money holdings in the PMP regime, as shown in Figure\ \ref{fig:ss_money} in the Appendix. This can be related to the general characteristics of the learning problem and aspects of deep reinforcement learning. The Taylor rule provides less feedback for low inflation values, and the household's utility function is relatively flat with respect to money holdings at the steady state values. This matters because episodes terminate when absolute changes in utility fall below the threshold $d^{min}_u$. This contributes to less precise learning of the low-inflation steady state in the current setting. This discrepancy is small in inflation terms\footnote{Of the order of 0.001 of net inflation in percentage points.} but gets amplified in the volatility of money holdings in the low-inflation steady state. Finally, steady state utility values in Table\ \ref{table:ss_conv} are also close to the long-run optima with mean deviation well below 1\% in all regimes. We next investigate in more detail the learning agent's actions going beyond mere convergence to steady state values.

\begin{figure}[h!]  
\begin{subfigure}{0.55\textwidth}
\hspace*{-1.cm}
\includegraphics[width=\linewidth]{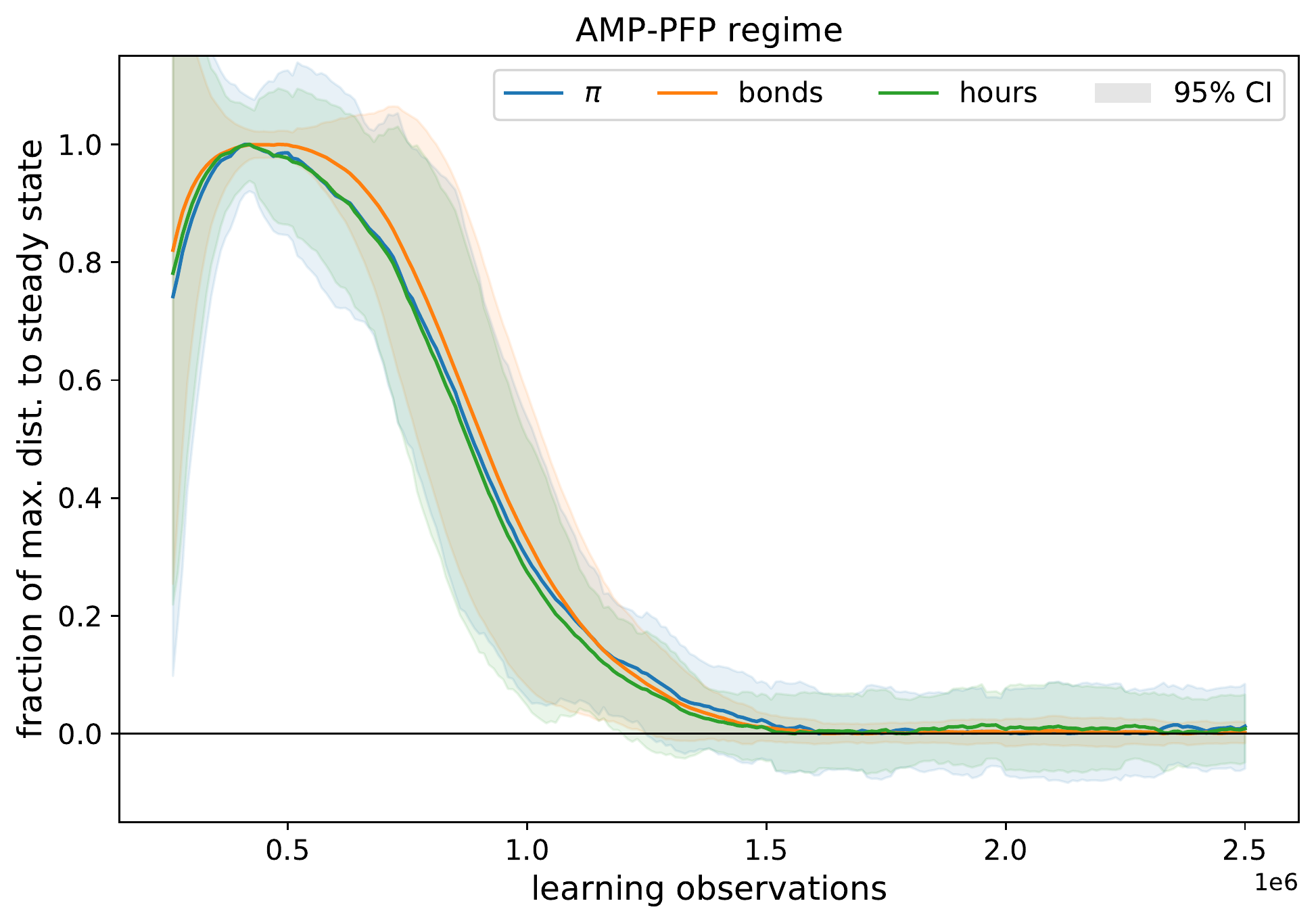}
\end{subfigure}
\begin{subfigure}{0.55\textwidth}
\hspace*{-1.cm}
\includegraphics[width=\linewidth]{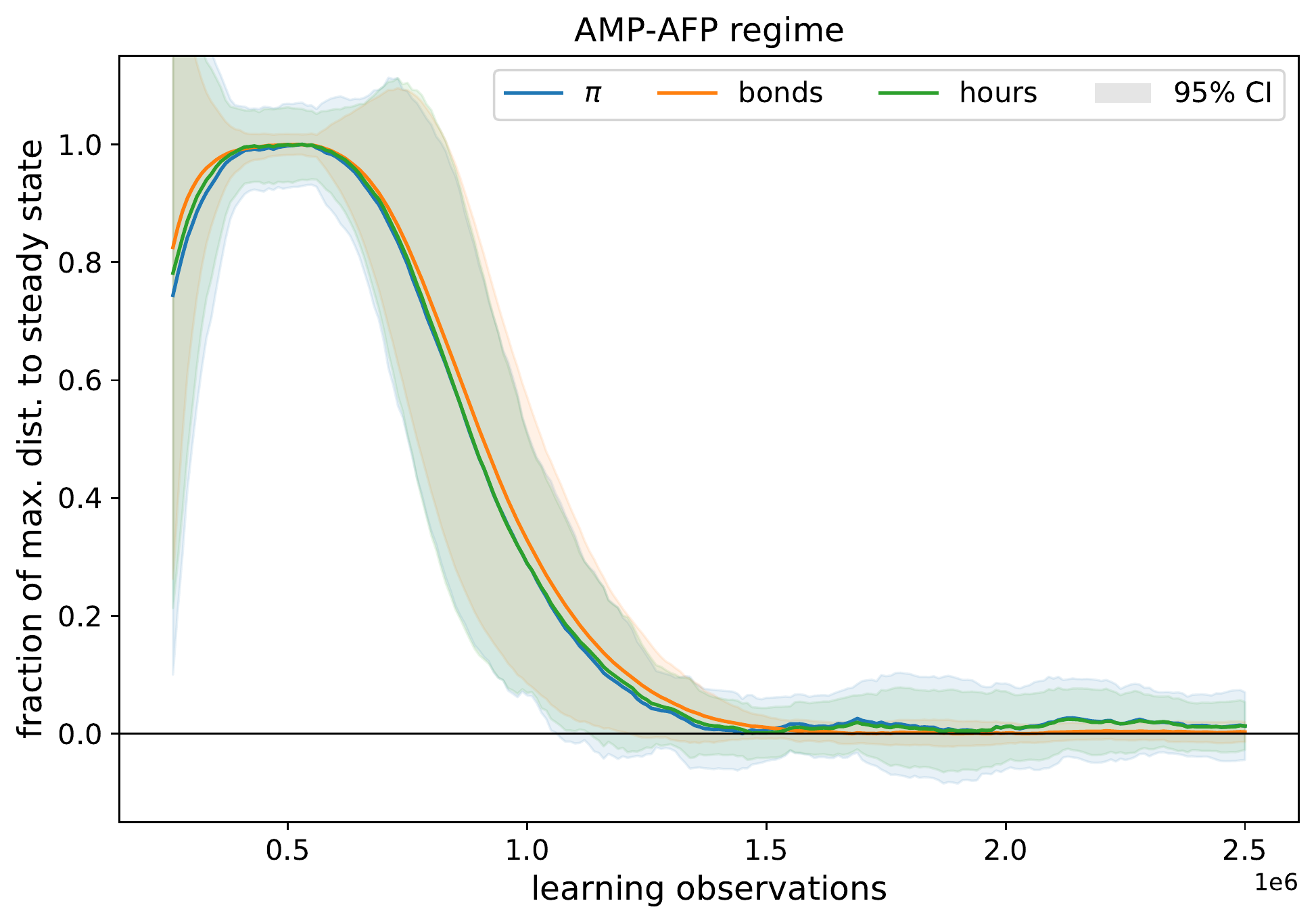}
\end{subfigure}

\begin{subfigure}{0.55\textwidth}
\hspace*{-1.cm}
\includegraphics[width=\linewidth]{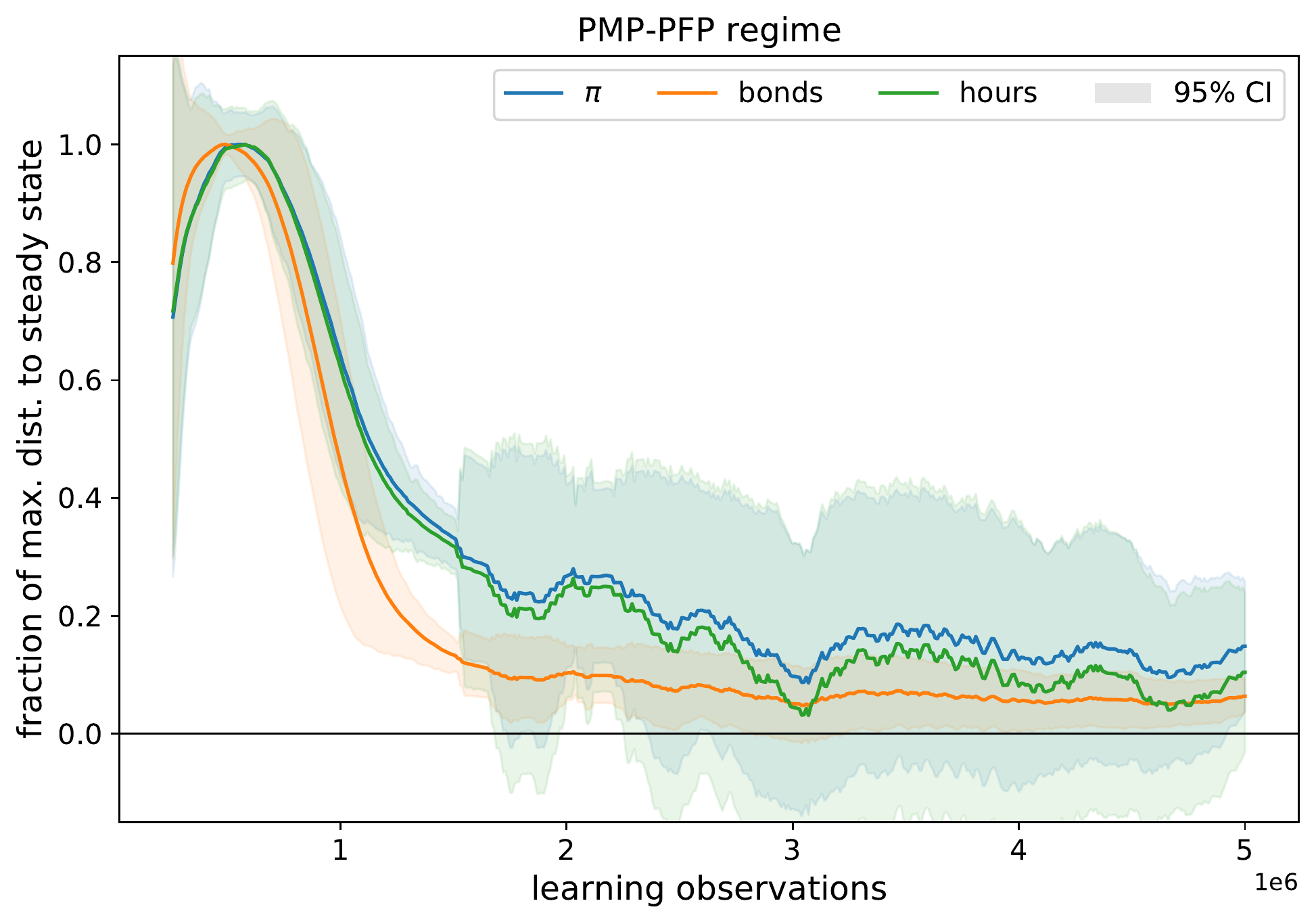}
\end{subfigure}
\begin{subfigure}{0.55\textwidth}
\hspace*{-1.cm}
\includegraphics[width=\linewidth]{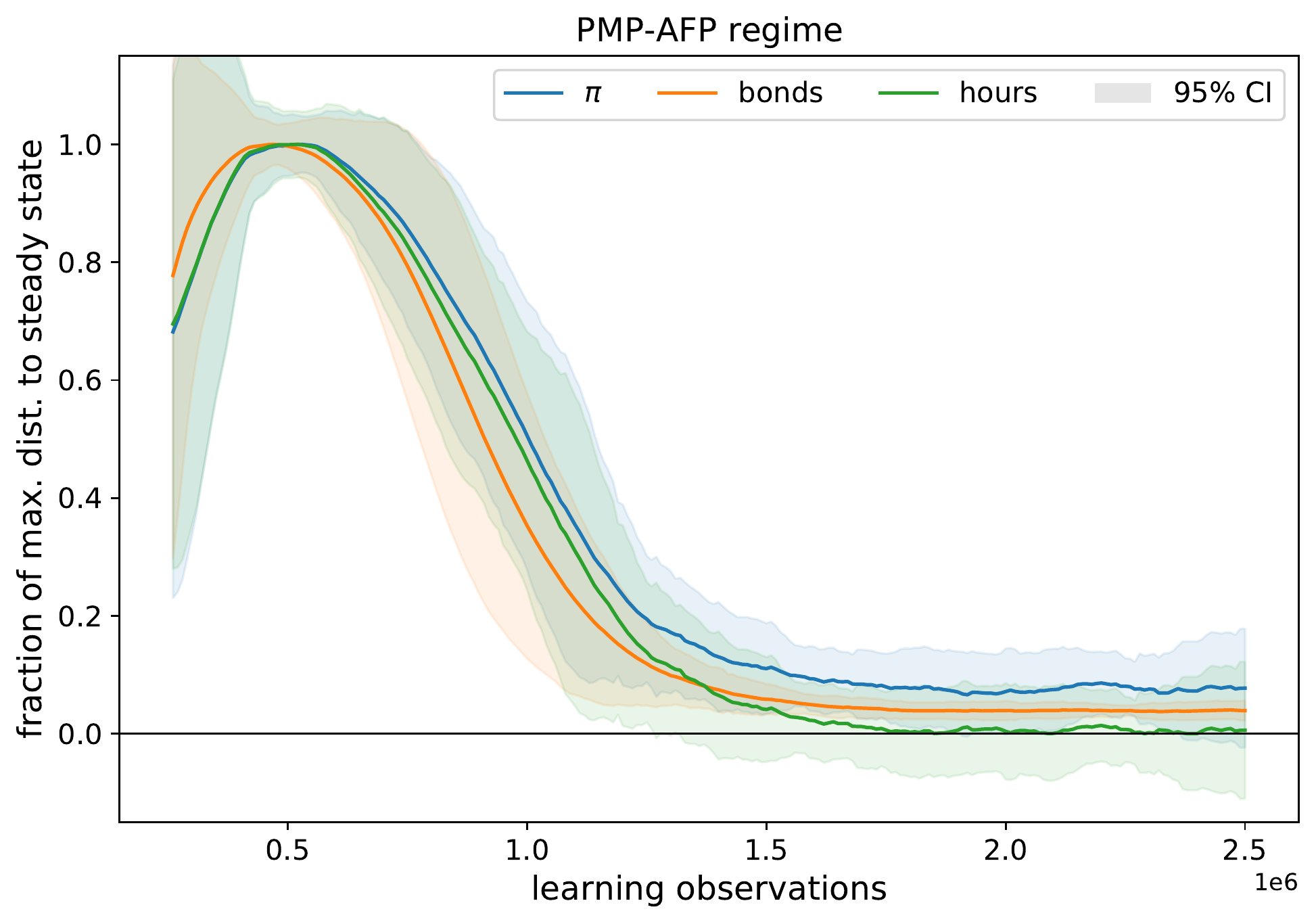}
\end{subfigure}
\caption{Comparison of steady state action convergence for different policy regimes at end of test episodes: monetary policy (rows) and fiscal policy (columns). Shaded areas show 95\% confidence intervals. Source: Authors' calculations.} \label{fig:ss_conv}
\end{figure}

\newpage

\subsection{Measuring bounded rationality}
\label{sec:rationality}

Household optimal or rational behaviour under deep reinforcement learning is taken to mean that the agent follows an Euler path, i.e. that its actions are in line with the first-order conditions (FOC) from Eq.\ \ref{eq:fisher}--\ref{eq:ls1}, and that it has learned about the steady state values of the model. The deep reinforcement learning agent does not know about either on the onset of learning but rather has to infer them to solve its reward maximisation problem. Hence, the FOC together with the household's actions and the realised values of state variables allow us to gauge the rationality of the agent at different stages of learning. We define the {\it FOC-distance} as
\begin{equation}\label{eq:foc_dist}
d^{FOC}_x\,\equiv\, \big|FOC(x)-1\big|\,,
\end{equation}
which allows us to evaluate deviations in a standardised way. A value of zero implies that the agent is on an Euler path, fulfilling its first order conditions. The explicit expression for the Euler equation (\ref{eq:fisher}), which we call the {\it Euler distance}, is 
\begin{equation}\label{eq:eul_dist}
d^{FOC}_{\pi}\,=\, \bigg|\beta\, \mathbb{E}_t\big[\big(\frac{c_{t+1}}{c_t}\big)^{-\sigma} \frac{R_t}{\pi_{t+1}}\big]-1\bigg|\,.
\end{equation}

We test FOC-learning by focusing again on the common AMP-PFP regime. Eq.\ \ref{eq:fisher}--\ref{eq:ls1} are evaluated analogously to the previous section by looking at the final transitions of each test episode at different stages of learning. Deep reinforcement learning does not provide us with explicit expectations. Instead, expected values in (\ref{eq:eul_dist}) are taken to be next-period realised values. That is, the agent's expectations are interpreted consistent with its actions, which is a simple form of self-fulling expectations.

Normalised FOC-learning curves are shown in Figure\ \ref{fig:FOC_AMP_PFP} for the Euler equation, money demand and labour supply. These are very similar to the convergence of household actions in Figure \ref{fig:ss_AMP_PFP}. The three learning phases can be clearly identified and coincide with convergence to the steady state values. This is in line with the GPI framework from Section\ \ref{sec:gpi}. There is a joint convergence of behaviour (FOC) and state learning. This means that the state values of Figures\ \ref{fig:ss_AMP_PFP} and \ref{fig:ss_conv} are the corresponding equivalents of Eq.\ \ref{eq:foc_dist} in the state space. Convergence in both corresponds to the rational expectation equilibrium. This also means that more generally, Eq.\ \ref{eq:foc_dist} in either the action or state space quantify {\it bounded rationality} in this class of models. For instance, maximal distances in Figure\ \ref{fig:FOC_AMP_PFP} correspond to random actions, while zero distances imply compliance with the FOC and convergence to a steady state equilibrium.

In adaptive learning, FOC-distances are zero by construction. However, bounded rationality can still be assessed within the same framework by looking at state convergence in Eq.\ \ref{eq:sslearning}, provided Proposition\ 3 holds. The FOC-learning curves for all policy regimes are given in Figure\ \ref{fig:foc_all} in the Appendix. These again converge to zero, or very close to zero within confidence bounds, in all four cases. The small discrepancies from steady state actions values are reflected in non-zero FOC-distances as suggested by GPI. Contrary to the learning of steady state actions, FOC-distances may converge partly asynchronously. However, we always observe the three phases of learning, and full convergence of all quantities coincides between state and action learning. We next look into how these learning dynamics can be used to arrive at testable results.

\begin{figure}[!ht] \centering%
\includegraphics[width=0.7\textwidth]{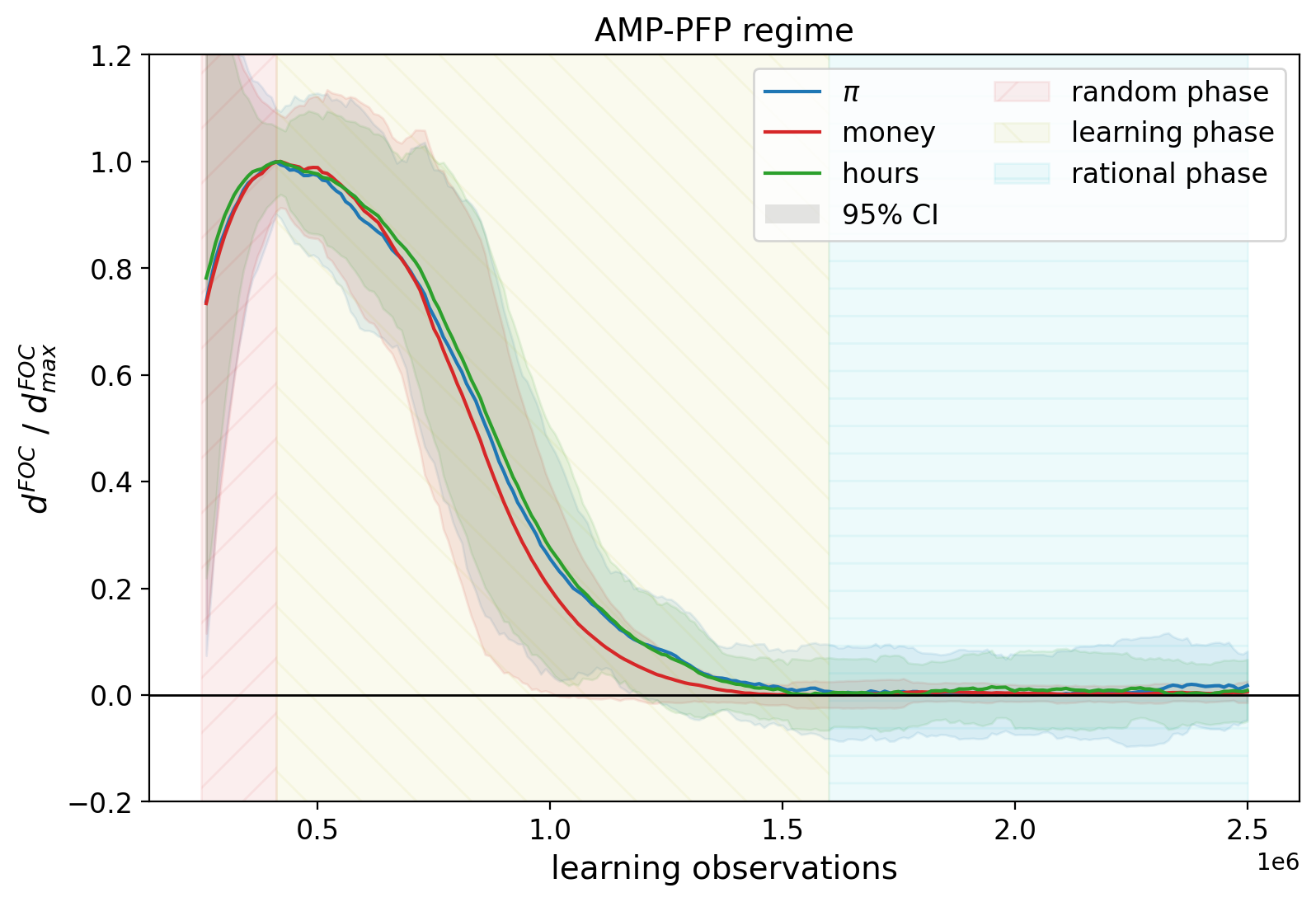} 
\caption{FOC-learning of housefold actions for the AMP-PFP regime at end of episodes. Shaded bands show 95\% confidence intervals. Hatched areas show learning phases. Source: Authors' calculations.} \label{fig:FOC_AMP_PFP} 
\end{figure}%

\newpage

\subsection{Agent behaviour as a model parameter}
\label{sec:expectations}

One of the contributions of this study is the development of tools to quantify bounded rationality. The FOC-distances (\ref{eq:foc_dist}) do this by measuring the difference in household behaviour through the choice of consumption, bond savings and hours worked compared to optimal value during learning. 
The basic idea is that the agent solves the model at any stage or learning, being `more or less rational.' Working with such an agent, one needs to determine the state of learning. Some evaluation criterion will be used for this, e.g. like FOC-distances or a calibration to data, such that the state of learning is a free parameter.  

Questions of interest are how to interpret these learning results, and, relatedly, how these can be brought to the data? The representative agent setting, and the fact that learning quantities themselves, like the steps to convergence, or the maximal FOC-distances during learning, depend on a variety of parameters of the model and learning algorithm, make the interpretation of results more difficult. This is even more so when we want to connect the learning setting to the real world. One possibility of interpretation is to interpret the state of learning, as expressed through steady state or FOC-distances, as the collective state of the population of agents in the economy represented by the agent at different stages of learning. The learning curves shown before offer devices to assess this state during learning. This allows to interpret learning results detached from the parameters of the model or learning algorithms used.

We demonstrate how such an approach can be used to address agent behaviour in the context of inflation expectations. Using the Euler distance, we run an experiment to investigate the relationship between the current interest rate $R_t$ and next-period inflation $\pi_{t+1}$ as determined by the household's consumption choice. The agent behaves consistently with its own model $\mathcal{P}_{\phi}$ based on the current state, so next-period inflation is equated with $\mathbb{E}_{t}[\pi_{t+1}]$. We fix the household's real consumption schedule at its optimal value, i.e. $c_t=c_{t+1}=c_{ss}$. This is achieved by setting $n_t=n_{ss}$ at all times, which fixes $c_t$ through the clearing of the goods market. That is, we isolate the learned relation between inflation and interest rates by fixing other actions at their optimal values. The Euler equation (\ref{eq:fisher}) then simplifies to the Fisher equation, i.e. a simple relation between this period's interest rate and next period's inflation,
\begin{equation}
\label{eq:fisher2}
\mathbb{E}_t[\pi_{t+1}]\,=\,\beta R_t\,.
\end{equation}

We next take the household agent at each test stage of learning\footnote{Algorithm \ref{algo:train_test} saves the household agent at each test loop at different stages of learning. We now reload the partially trained agent.} and rerun all test cycles with hours fixed as described above recording all state transitions. The results of this exercise are depicted in Figure\ \ref{fig:exp_AMP_PFP}. Both axes show net inflation expectations, the horizontal axis the ones implied by the current rate of interest $R_t$ and the Fisher equation (\ref{eq:fisher2}), and the vertical axis the actual agent expectations measured through the household's consumption choices. The dotted diagonal line describes rational behaviour, i.e. following FOC, again the Fisher equation. The household's actual actions during each test transition are given by the scatter points at different times of learning as indicated by the colour coding. The agent's learning curve is traced out by the dashed line. The vertical distance between this line and the diagonal measures the average deviation of expectations from the rational expectation equilibrium alongside this dimension, that is, bounded rationality.

We can draw the following conclusions from this experiment. The household's initial inflation expectations deviate about 5-7\% from optimal expectations (random phase, purple dots). During learning the household's expectations and actions converge to the optimal values, i.e. the vertical distance between the two lines narrow (learning phase, blue dots). Eventually, agent actions coincide with the Fisher equation (rational phase, yellow dots).

This experiment shows how we can arrive at tangible, i.e. testable, propositions from the deep reinforcement learning framework. This type of analysis could now be used to bring a model to the data by quantifying  real-world agents' state of learning or their bounded rationality. FOC-distances could be estimated from suitable datasets, and agents at the corresponding learning stage could then be used for further analyses, like counterfactual experiments and their outcomes be compared to the conventional case of fully rational agents. We leave this to future work.

\begin{figure}[!ht] \centering%
\includegraphics[width=0.7\textwidth]{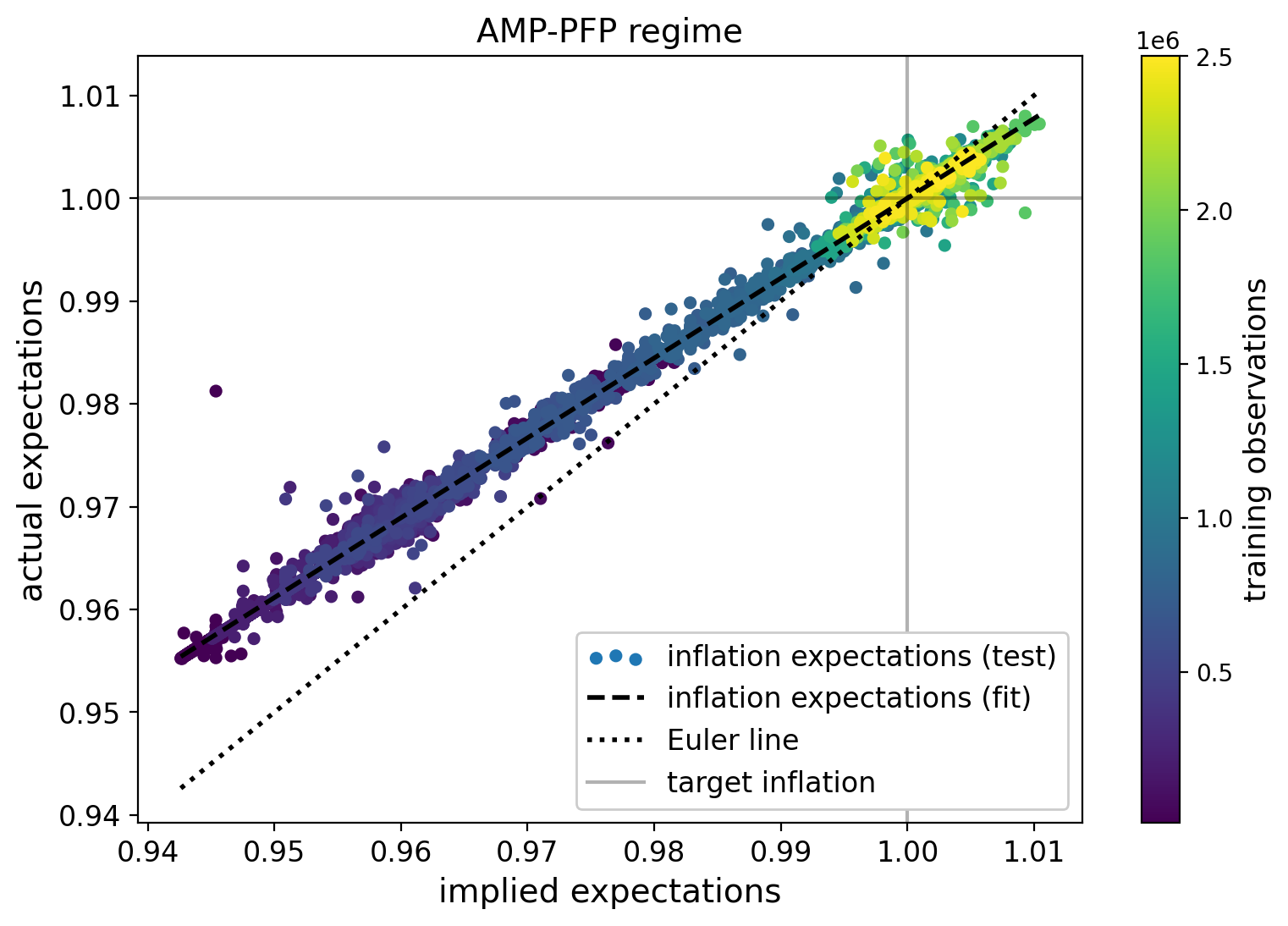} 
\caption{Evaluation of Euler equation (\ref{eq:fisher}) with consumption schedule, i.e. hours worked, fixed at the corresponding steady state value (Fisher equation) for AMP-PFP policy regime. Source: Authors' calculations.} \label{fig:exp_AMP_PFP} 
\end{figure}%

\newpage

\subsection{Learning under uncertainty}
\label{sec:uncertainty}

The analyses so far have not taken uncertainty into account. However, it is an interesting question whether shocks to monetary or fiscal policy, or to the labour supply, affect the household agent's ability to learn in deep reinforcement learning, and if so, whether there are differences in outcomes. We repeat the experiments from Section\ \ref{sec:learnability} for the AMP-PFP and the PMP-AFP regimes, with the shock sizes for both scenarios given in Table\ \ref{tab:para}. These are comparably small for the active policy regime but large relative to steady state net inflation in the low-inflation state. The agent again follows the learning protocol in Algorithm\ \ref{algo:train_test}, with the difference that now shocks are realised at each state transition, and that shocks enter the state observations.\footnote{The state dimension was not reduced in previous experiment without shocks, but rather $\epsilon^{\tau}_t$, $\epsilon^{R}_t$, $\epsilon^{y}_t$ where set to their respective means of zero, one and one. Note that constant terms do not affect the learning outcome as there is no variation entering the optimisation process.}

The results for these experiments are summarised in Figure\ \ref{fig:learning_wshocks}, where we again look at end-of-episode transitions during testing.\footnote{We take the second last transition here as the last transition contains a shock to next period's monetary policy which is not relevant if an episodes ends.} Convergence again means actions and FOC-distances are in line with those for the cases without shocks and converge to the rational expectation equilibrium. The only difference in learning is that the confidence intervals are generally wider. This is expected in the presence of external random fluctuations of state variables. Thus, the presence of shocks does generally not inhibit learning.

An interesting aspect of learning in the presence of shocks is that solutions correspond to stochastic steady states, which are generally hard to assess.
\begin{figure}  
\begin{subfigure}{0.55\textwidth}
\hspace*{-1.cm}
\includegraphics[width=\linewidth]{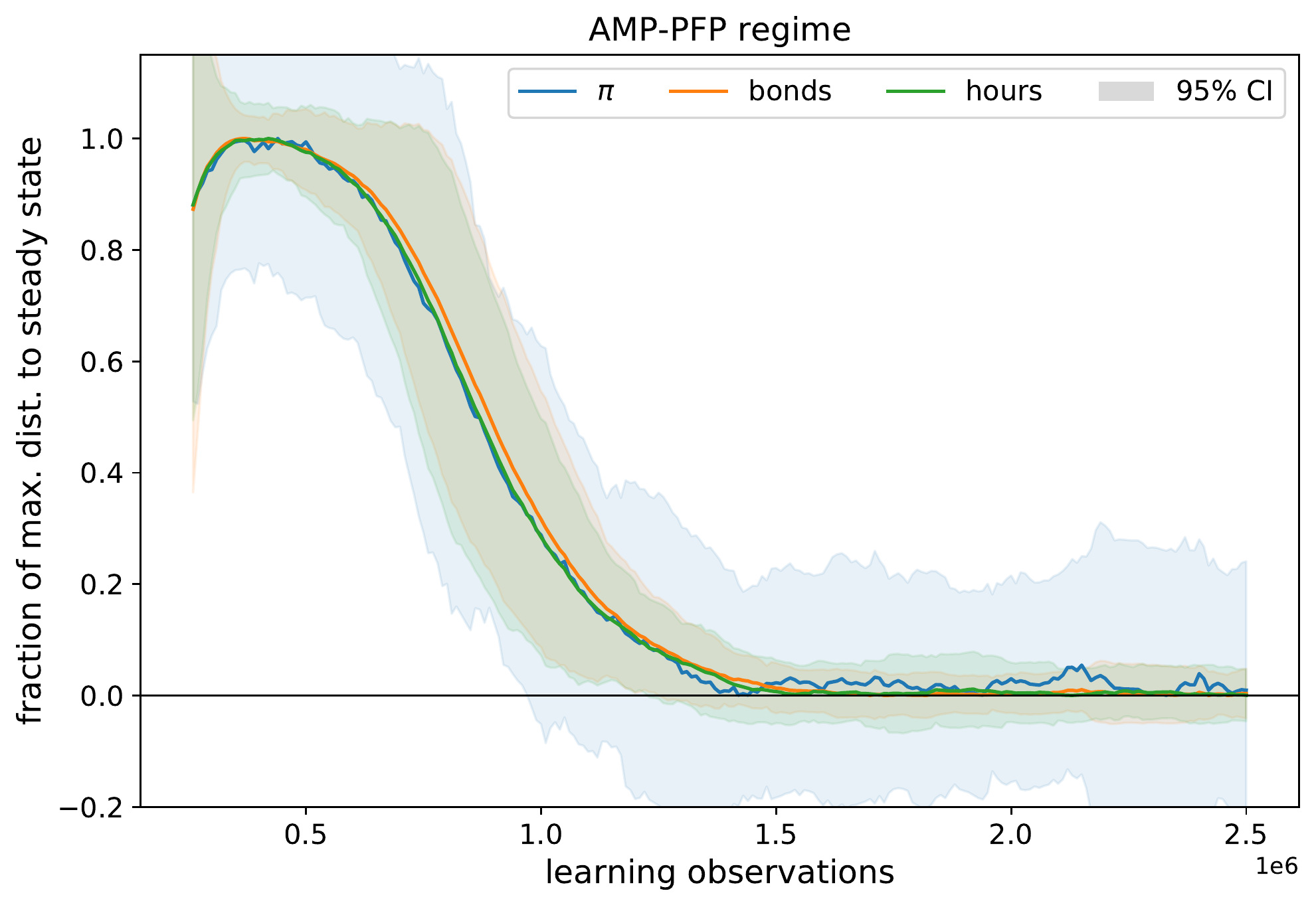}
\end{subfigure}
\begin{subfigure}{0.55\textwidth}
\hspace*{-1.cm}
\includegraphics[width=\linewidth]{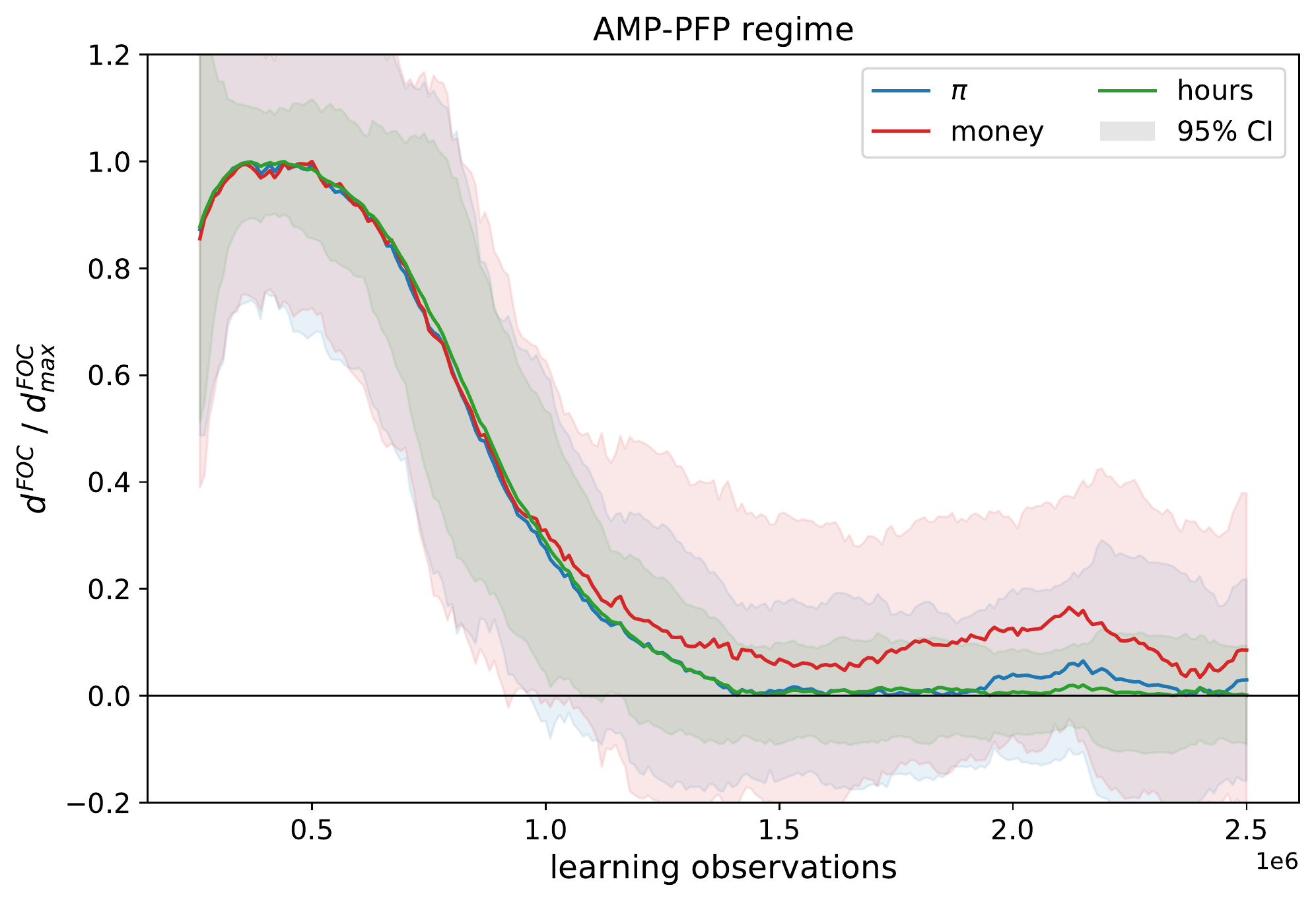}
\end{subfigure}

\begin{subfigure}{0.55\textwidth}
\hspace*{-1.cm}
\includegraphics[width=\linewidth]{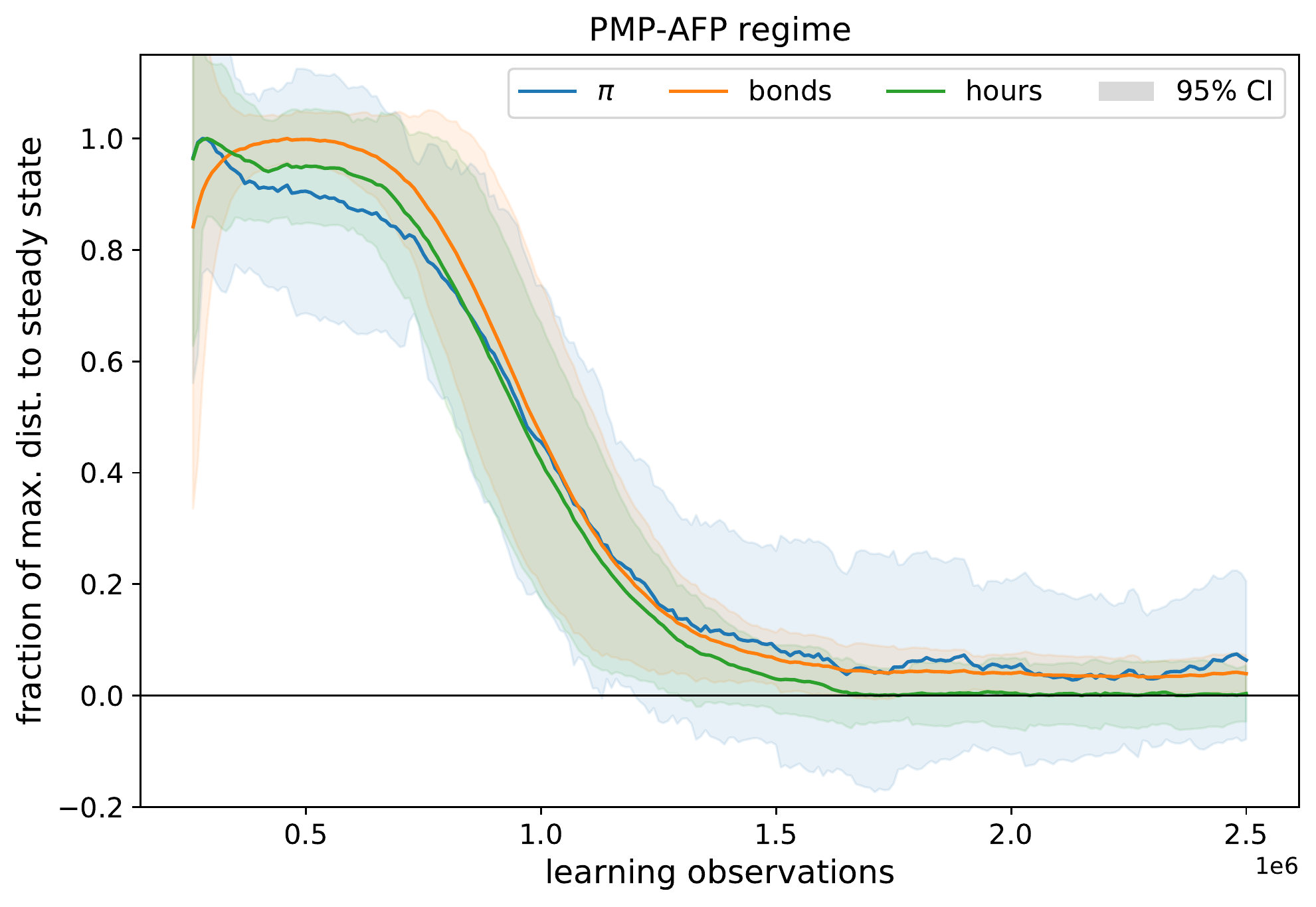}
\end{subfigure}
\begin{subfigure}{0.55\textwidth}
\hspace*{-1.cm}
\includegraphics[width=\linewidth]{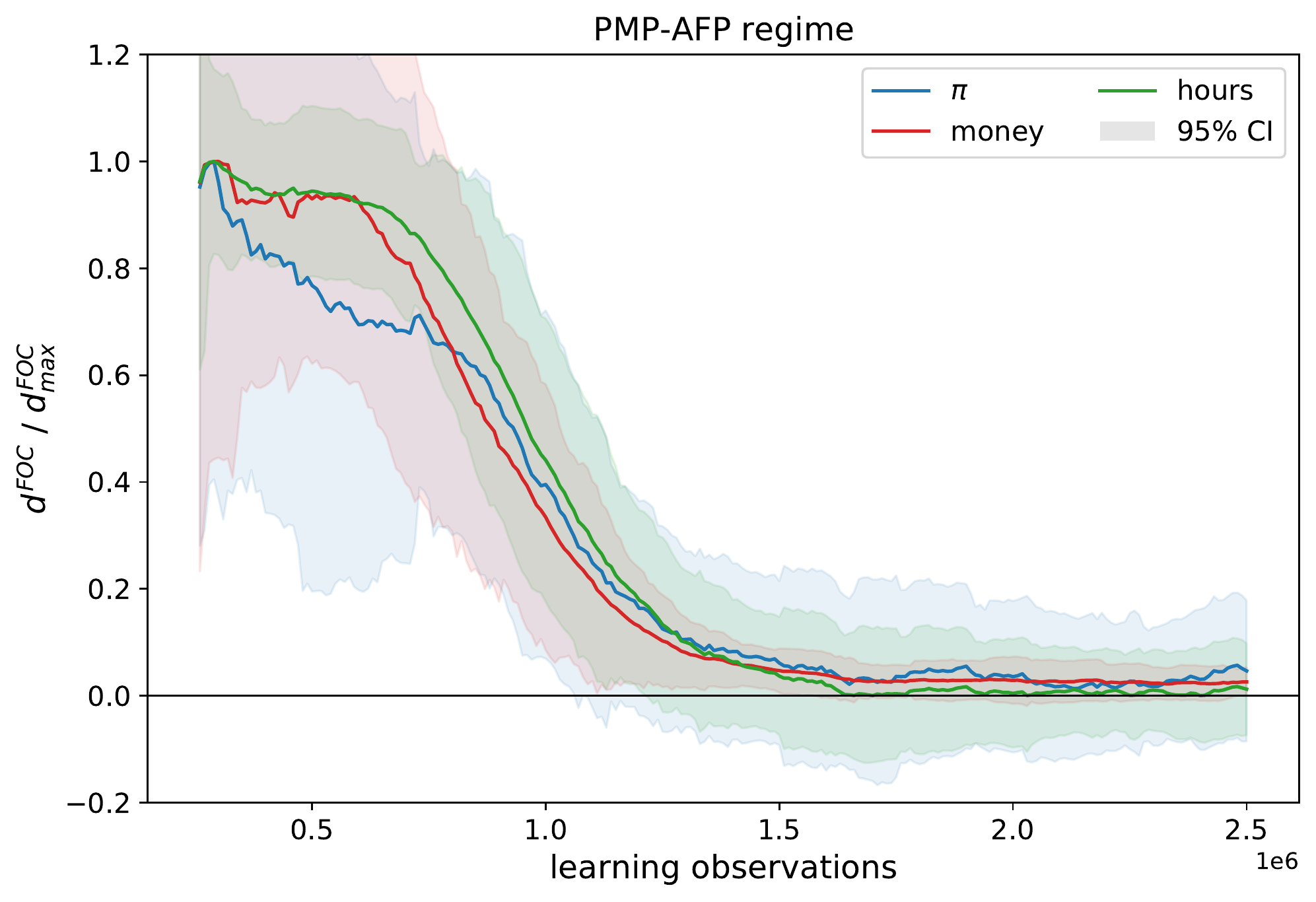}
\end{subfigure}
\caption{Evaluation of steady state (left column) and FOC-distances (right column) learning curves in the presence of shocks within the AMP-PFP (upper row) and PMP-AFP (lower row) policy regimes for end-of-period test transitions. Source: Authors' calculations.} 
\label{fig:learning_wshocks}
\end{figure}

\newpage

\subsection{Learning challenges}
\label{sec:challenges}

Despite the general learnability results we presented, there are significant challenges to the application of deep reinforcement learning. The sensitivity to hyperparameter choices, or {\it brittleness} of deep and, in particular, deep reinforcement learning is known in the literature \citep{Haarnoja2018sac,Heaven2019brittle}. Some of this also affects our analysis. We discuss the main learning parameters relevant to the present analysis, that is, those factors which lie at the intersection between learning and the economic model setting, and that have a significant effect on the learning outcome. The aim of this is to lay out challenges for future research on the one hand, and, on the other, to facilitate the adaptation of deep reinforcement learning to economics and finance problem more generally.

Three hyperparameters relevant to our learning results are the speed of learning, action and state bounds, and episodic termination criteria:

\subsubsection{The speed of learning}
This is controlled by the {\it learning rate} $\alpha_{learn}$ which sets the step size by which the parameters in the agents functions $\mathcal{P}_{\phi}$ and $Q_{\theta}$ are updated during learning.\footnote{Let $Q(\theta)$ be a differentiable function and $\nabla_{\theta}Q$ its derivative. A gradient descent step during optimisation, e.g. to solve a least-square problem as in (\ref{eq:bell_resid}), takes the form $\theta\leftarrow\theta-\alpha_{learn}\nabla_{\theta}Q$; $\alpha_{learn}$ determines the size of this update.} The larger its values, the faster the agent learns, and the less steps it will need to converge. So, ideally $\alpha_{learn}$ is chosen as large as possible. However, too large a value destabilises learning, leading either to explosive or stagnant behaviour. Our baseline value $\alpha_{learn}=10^{-5}$ is relatively small. However, larger values led to unstable learning and breakdown of the optimisation algorithm, highlighting the challenges of our learning problem.

\subsubsection{Agent experience}
Ideally one would offer a learning agent the full state space to learn from. Unfortunately, this is quite challenging given the current setting and will likely lead to a failure to learn. For instance, the agent may end up converging to a mid-point between the two inflation steady states. In the current setting, learning outcomes are often best if the agent's action space is centred around one of the two steady states. This can be seen in Figure\ \ref{fig:conv_sym_lowCD}, which shows convergence to steady state action values in the PMP-AFP regime. Convergence is improved relative to the non-centred case in the lower right of Figure\ \ref{fig:ss_conv}. Note, however, that comparable improvement is not seen for the PMP-PFP regime, such that this is not a universal remedy to convergence issues. This is not the only parameter to improve learning in this regime. 

\subsubsection{Episode termination}
A feature of our learning problem is that optimal solutions are steady states, i.e. points in the action, state and utility space which map into themselves. This is accounted for by the episodic termination criterion using $d_{u}^{min}$ in our setting. This does not contain direct information about where the optimal solution of the household's problem lies. A large value of $d_{u}^{min}$ lets the agent have more experience as episodes terminate and restart after fewer steps. However, this may lead to decreased learning precision around steady states. On the other hand, a small value increases precision, but offers less experience to the agents as episodes take more steps and the agents potentially spends more time in `uninteresting' regions of the state space. The effect of increased precision on action convergence in the PMP-AFP regime is show on the RHS of Figure\ \ref{fig:conv_sym_lowCD}.  Here, $d_{u}^{min}$ has been set to 1e-9 instead of 1e-7.\footnote{A technical aspect to consider when setting a low $d_{u}^{min}$ is potentially increased memory need if episodic transitions are saved during training or testing.} Convergence again improves compared to the baseline.

\begin{figure}  
\begin{subfigure}{0.55\textwidth}
\hspace*{-1.cm}
\includegraphics[width=\linewidth]{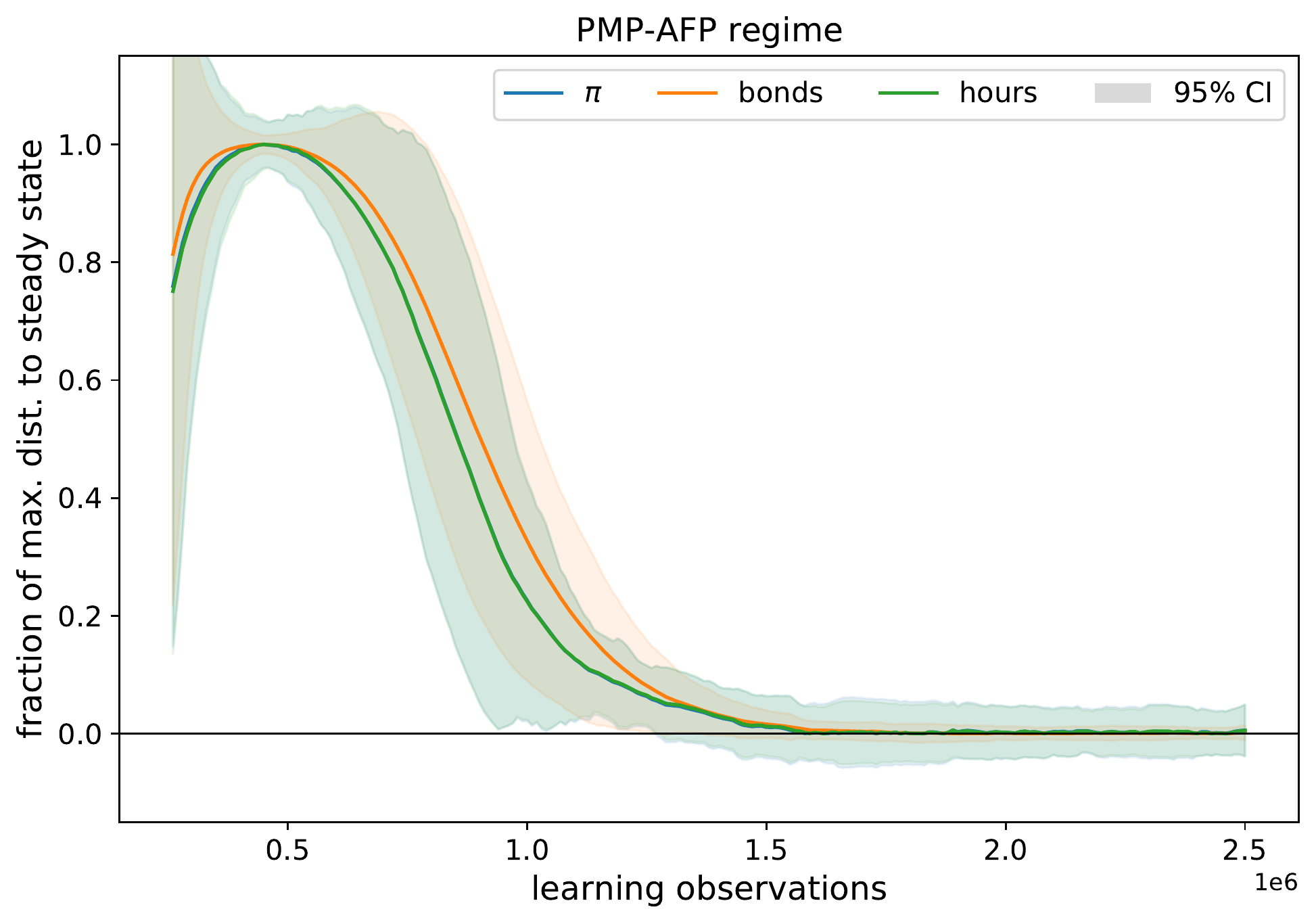}
\end{subfigure}
\begin{subfigure}{0.55\textwidth}
\hspace*{-1.cm}
\includegraphics[width=\linewidth]{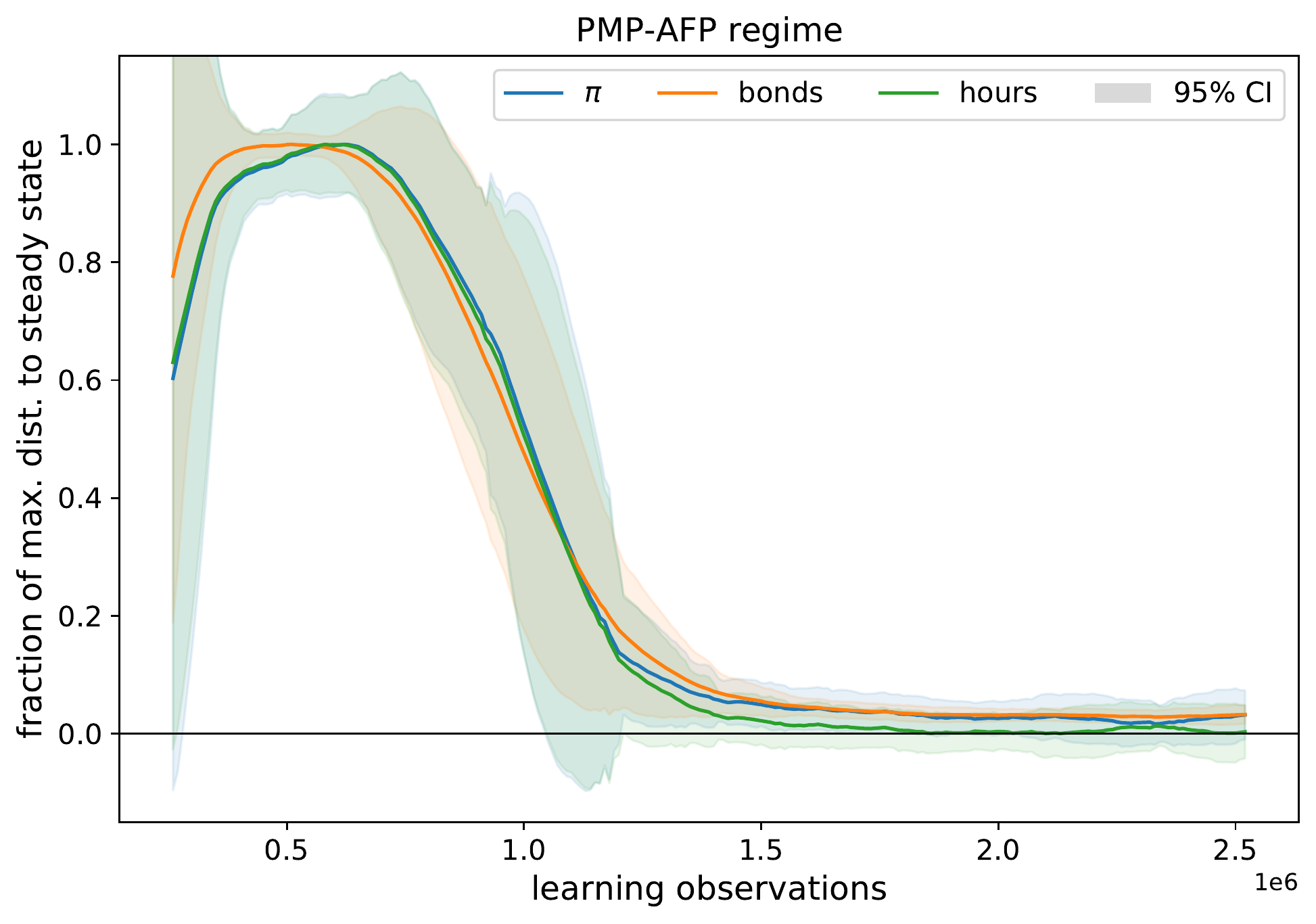}
\end{subfigure}
\caption{Learning convergence of household actions in the PMP-AFP regime. LHS: Symmetric action space around $\pi_L$,  $c^{act}\in[1.0000,1.0028]$. RHS: decreased convergence distance $d^{min}_u=$1e-9. Shaded areas show 95\% confidence intervals. Source: Authors' calculations.} \label{fig:conv_sym_lowCD}
\end{figure}

\section{Conclusion\label{sec:conclusion}}
We propose the use of deep reinforcement learning to solve dynamic models commonly used in economic analysis, particularly DSGE models. Here, agents are capable of solving the model based only on state transitions, knowing their preferences in the form of a reward function, but not knowing any other part of the structure of the model economy. We apply deep reinforcement learning to a monetary model with a representative household agent. The salient feature of this model, which is commonly used in the learning literature, is the existence of two steady states, where one of them may be interpreted as a liquidity trap with lower inflation and higher money holdings. The interaction between monetary and fiscal policy regimes defines the learnability of solutions in well known adaptive learning approaches, where this then serves as a selection criterion for the study and plausibility of a regime. By contrast, the deep reinforcement learning agent recovers all possible steady state solutions under the different policy regimes. Given that this learning behaviour is more general, as seen through the lens of generalised policy iteration, this widens the scope of solutions of interest, irrespective of their local dynamics. This is achieved by the ``global''  nature of deep reinforcement learning without the need of a local model approximation. Agents' learning and behaviour is ultimately determined by the long-run reward subject to  constraints. In our case this consists of household utility subject to the standard inter-temporal accounting constraint.

Conceptually, deep reinforcement learning allows us to define and measure bounded rationality on a continuous spectrum of agent behaviour, which is limited by the rational expectations equilibrium as the learning goal. In other words, agent behaviour becomes a free parameter of the model. We proposed measures and procedures to assess this. These can be used in either computational exercises or to quantify real-world agent behaviour. For instance, the spectrum of learning given by state value or first-order learning curves, can be used to calibrate agent behaviour.

However, the deep reinforcement learning approach is not without challenges. Learning approaches based on deep artificial neural networks can be brittle and sensitive to parameter choices. We discussed some of the choices, like the learning rate, action and state space bounds or episodic termination criteria. Furthermore, we only considered the single-agent case in this study while it is known that multi-agent settings can be considerably more difficult to handle with these techniques. Deep reinforcement learning also requires substantial computational resources and specialised programming skills.

However, agent behaviour is one of the most contentious parts of modern economic theory and analysis. This criticism goes beyond the DSGE plus rational expectations modelling paradigm discussed here. Deep reinforcement learning, on the other hand, is a general purpose technology and can be applied to a wide range of problems. The only  requirement is that these problems can be moulded into its fairly general structure. As such, we believe that deep reinforcement learning offers a flexible and potentially highly useful tool to address major conceptual and practical concerns, and to interlink different approaches within a rich research agenda.

\clearpage
\bibliographystyle{agsm}
\bibliography{learning}
\clearpage

\vfill\pagebreak

\appendix
\numberwithin{equation}{section}
\numberwithin{figure}{section}
\numberwithin{table}{section}
\small

\section{Technical Appendix}
\label{sec:app}

\subsection{Derivation of \eqref{linear_BK}\label{appendix1}}
\label{sec:app_bk}

In an neighbourhood of a non-stochastic steady state $\pi$ and $c$, we can derive a linear approximation
\begin{align}
&\text{Euler Equation:\ \ \ \ \ \  \ \ \ \ }  \hat  R_t= \beta^{-1} E_t \hat \pi_{t+1}+\frac{\sigma}{\beta}\frac{\pi}{c} (E_t \hat c_{t+1}-\hat c_t) \label{eq:linear_ee}\\
&\text{Monetary Policy:\ \ \ \ \ \  \ \ } \hat R_t=\alpha \hat \pi_t+\delta \hat \varepsilon^R_t \ \ \ \text{where $\alpha=f'(\pi)$ and $\delta=f(\pi)$}    \label{eq:linear_mp} \\
&\text{Fiscal Policy \& GBC:\ \  \ }  \hat b_t+\hat m_t+\hat \varepsilon_t^\tau= (\frac{1}{\beta}-\gamma)  \hat b_{t-1} -\frac{m+R b}{\pi^2} \hat\pi_t +\frac{1}{\pi}\hat m_{t-1} + \frac{b}{\pi} \hat R_{t-1} \label{eq:linear_gbc}\\
&\text{Money Demand:\ \ \ \ \ \ \ \ \ \  }  \hat m_t= \frac{m}{c}\hat c_t -\frac{1}{\sigma} \frac{m}{R(R-1)} \hat R_t \label{eq:linear_m} \\ 
&\text{Output:\ \ \ \ \ \ \ \ \ \ \ \ \ \ \ \ \ \ \ \ \ }  \frac{\sigma(1-\eta)+\eta+\varphi}{1+\varphi}\frac{1}{c}\hat c_t= \hat \varepsilon^y_t  \label{eq:linear_y}
\end{align}
Note that $\hat x_t$ denotes the deviation of variable $x_t$ from steady state. We now consider determinacy of the linearized system. We can rewrite \eqref{eq:linear_ee}--\eqref{eq:linear_y} as a bivariate forward-looking system of the form

\begin{equation}
\begin{bmatrix}
\hat  \pi_t \\
\hat b_t
\end{bmatrix}=
\begin{bmatrix}
B_{11} & B_{12} \\
B_{21}  &B_{22}
\end{bmatrix}
\begin{bmatrix}
\hat E_t \pi_{t+1} \\
\hat E_t b_{t+1}
\end{bmatrix}+
\begin{bmatrix}
C_{11} & C_{12} & C_{13}\\
C_{21}  &C_{22}& C_{23}
\end{bmatrix}
\begin{bmatrix}
\hat  \varepsilon^R_t  \\
\hat  \varepsilon_{t}^\tau \\
\hat \varepsilon^y_t
\end{bmatrix}
\label{linear_BK}
\end{equation}
According to \cite{BK1980}, the solution to \eqref{linear_BK} is locally unique if and only if one eigenvalue is within the unit circle and the other eigenvalue is outside the unit circle. To assess this we bring the above expressions into an explicit form,
\begin{equation}
\begin{bmatrix}
 \frac{b\alpha}{\pi}-\frac{m\alpha}{\pi\sigma R(R-1)}&\frac{1}{\beta}-\gamma \\
\alpha&0
\end{bmatrix}
\begin{bmatrix}\hat{\pi} _t\\\hat{b} _t
\end{bmatrix}=
\begin{bmatrix}
\frac{m+\frac{1}{\beta}\pi b}{\pi^2}-\frac{m\alpha}{\sigma R(R-1)} & 1\\
\frac{1}{\beta}&0
\end{bmatrix}
\begin{bmatrix}E_t\hat{\pi} _{t+1}\\E_t\hat{b} _{t+1}
\end{bmatrix}+
\begin{bmatrix}
\frac{m\delta}{\pi\sigma R(R-1)}-\frac{b\delta}{\pi}&0&-\frac{m}{\pi c \xi}\\
-\delta&0&-\frac{\sigma\pi}{\beta c \xi}
\end{bmatrix}
\begin{bmatrix}
\hat{\varepsilon}^R _t\\
\hat  \varepsilon_{t}^\tau \\
\hat{\varepsilon}^y _t
\end{bmatrix}
\end{equation}
$\xi=\frac{\sigma(1-\eta)+\eta+\varphi}{(1+\varphi)c}$.
Therefore, 
\begin{equation}
\begin{bmatrix}
B_{11}&B_{12}\\
B_{21}&B_{22}
\end{bmatrix}=
\begin{bmatrix}
 \frac{b\alpha}{\pi}-\frac{m\alpha}{\pi\sigma R(R-1)}&\frac{1}{\beta}-\gamma \\
\alpha&0
\end{bmatrix}^{-1}
\begin{bmatrix}
\frac{m+\frac{1}{\beta}\pi b}{\pi^2}-\frac{m\alpha}{\sigma R(R-1)} & 1\\
\frac{1}{\beta}&0
\end{bmatrix}
\end{equation}
\begin{equation}
\begin{bmatrix}
C_{11}&C_{12}&C_{13}\\
C_{21}&C_{22}&C_{23}
\end{bmatrix}=
\begin{bmatrix}
 \frac{b\alpha}{\pi}-\frac{m\alpha}{\pi\sigma R(R-1)}&\frac{1}{\beta}-\gamma \\
\alpha&0
\end{bmatrix}^{-1}
\begin{bmatrix}
\frac{m\delta}{\pi\sigma R(R-1)}-\frac{b\delta}{\pi}&0&-\frac{m}{\pi c D}\\
-\delta&0&-\frac{\sigma\pi}{\beta c D}
\end{bmatrix}
\end{equation}
The two eigenvalues are $\frac{1}{\alpha\beta}$ and $\frac{1}{1/\beta-\gamma}$. Then a unique solution takes the form:
\begin{equation}
\begin{bmatrix}
\hat  \pi_t \\
\hat b_t
\end{bmatrix}=
\begin{bmatrix}
D_{11} & D_{12} & D_{13}\\
D_{21}  &D_{22}& D_{23}
\end{bmatrix}
\begin{bmatrix}
\hat  \varepsilon^R_t  \\
\hat  \varepsilon_{t}^\tau \\
\hat  \varepsilon^y_t
\end{bmatrix}
\label{linear_solution}
\end{equation}
%


\FloatBarrier
  
\subsection{Deep reinforcement learning parameterisation}
\label{sec:app_para}
 

\begin{table}[hbt!]
\centering
\resizebox{0.9\textwidth}{!}{
\begin{tabular}{l|ccl}
\toprule
parameter & AMP ($\pi^*$) & PMP ($\pi_L$) & description \\
\hline
\hline
 \multicolumn{4}{c}{{\it action bounds}} \\
\hline
$c^{act}_{min}$     & 1.005  & 1.000  & minimal consumption choice\\
$c^{act}_{max}$     & 1.015  & 1.003  & maximal consumption choice\\
$b^{act}_{min}$     & 4.000  & 3.965  & minimal bond holdings\\
$b^{act}_{max}$     & 4.080  & 4.045  & maximal bond holdings\\
$n_{min}$           & 0.990  & 0.990  & minimal hours worked\\
$n_{max}$ 		    & 1.010  & 1.010  & maximal hours worked\\
\hline
 \multicolumn{4}{c}{{\it initial state bounds}} \\
\hline
$m_{min}$           & 1.670  & 2.010  & minimal money holdings\\
$m_{max}$           & 1.750  & 2.110  & maximal money holdings\\
$b_{min}$           & 3.960  & 3.960  & minimal bond holdings\\
$b_{max}$           & 4.040  & 4.040  & maximal bond holdings\\
$c_{min}$           & 0.995  & 0.997  & minimal consumption\\
$c_{max}$           & 1.005  & 1.003  & maximal consumption\\
$\pi_{min}$         & 1.005  & 1.000  & minimal inflation\\
$\pi_{max}$         & 1.015  & 1.003  & maximal inflation\\
$n_{min}$           & 0.990  & 0.990  & minimal hours worked\\
$n_{max}$           & 1.010  & 1.010  & maximal hours worked\\
\hline
 \multicolumn{4}{c}{{\it learning algorithm}} \\
\hline
$\alpha_{learn}$	& 1.0e-5 & 1.0e-5  & learning rate \\
$d^{min}_u$         & 1.0e-7 & 1.0e-7  & utility difference (episode termination)\\
$\tau_{learn}$	    & 1.0e-3 & 1.0e-3  & target smoothing coefficient \\
$N_{train}$			& 2.5e6  & 2.5e6   & training steps (experiment) \\
$N_{interval}$		& 1.0e4  & 1.0e4   & training steps (between test episodes)\\
$N_{test}$			& 10	 & 10      & number of test episodes between training intervals\\
$N_{epi}^{max}$     & 2.5e4  & 2.5e4   & max. steps / episode (training or testing)\\
$N_{burn}$			& 1.0e4  & 1.0e4   & initial burn-in random actions\\
$N_{mem}$			& 2.5e4  & 2.5e4   & max. memory of state transitions\\
$N_{batch}$			& 256	 & 256     & batch size for parameter updates\\
$N^{hidden}_{layers}$& 2	 & 2       & number of hidden layers in $\mathcal{P}_{\phi}$\\
$N^{hidden}_{nodes}$& 32	 & 32      & number of nodes in hidden layers of $\mathcal{P}_{\phi}$\\
\bottomrule
  \end{tabular}
  }
\caption{Learning parameters following \cite{Haarnoja2018sac} for active monetary policy (AMP) around $\pi^*$ and passive monetary policy (PMP) around $\pi_L$. Action bounds refer to the minimal and maximal actions the agent can choose from. Initial state bounds refer to the min./max. values of initial state variables episodes can be sampled from. Learning algorithm parameters relate to parts of the agent's optimisation process which are not directly related to economic quantities (state variables). Additional learning parameters are the number of policy $\mathcal{P}_{\phi}$ and critique $Q_{\theta}$ updates per step (1), and the use of automatic entropy tuning (true).}
\label{tab:para_2}
\end{table}
\FloatBarrier


\subsection{Additional learning results}
 
\begin{figure}[b]  
\begin{subfigure}{0.55\textwidth}
\hspace*{-1.cm}
\includegraphics[width=\linewidth]{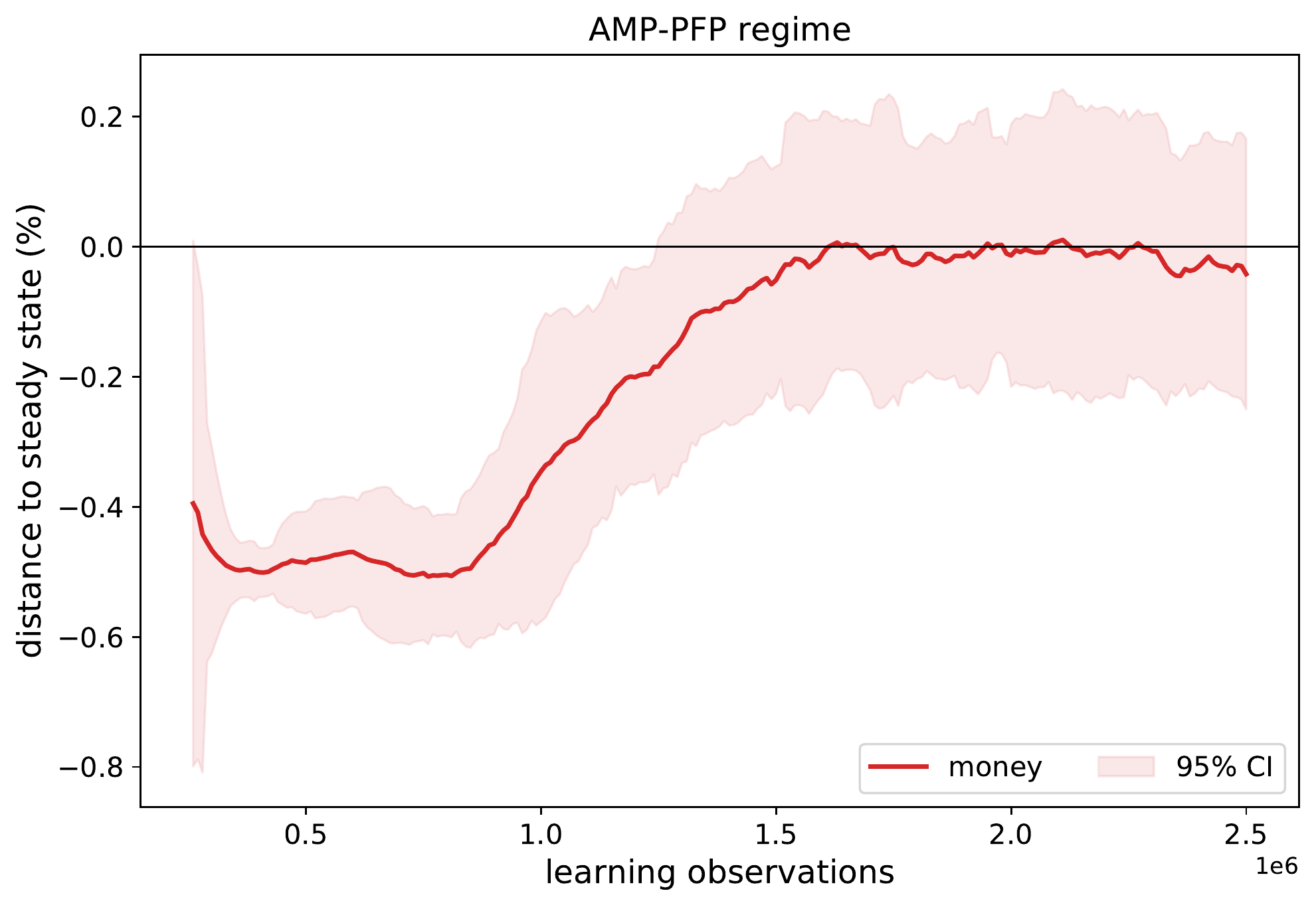}
\end{subfigure}
\begin{subfigure}{0.55\textwidth}
\hspace*{-1.cm}
\includegraphics[width=\linewidth]{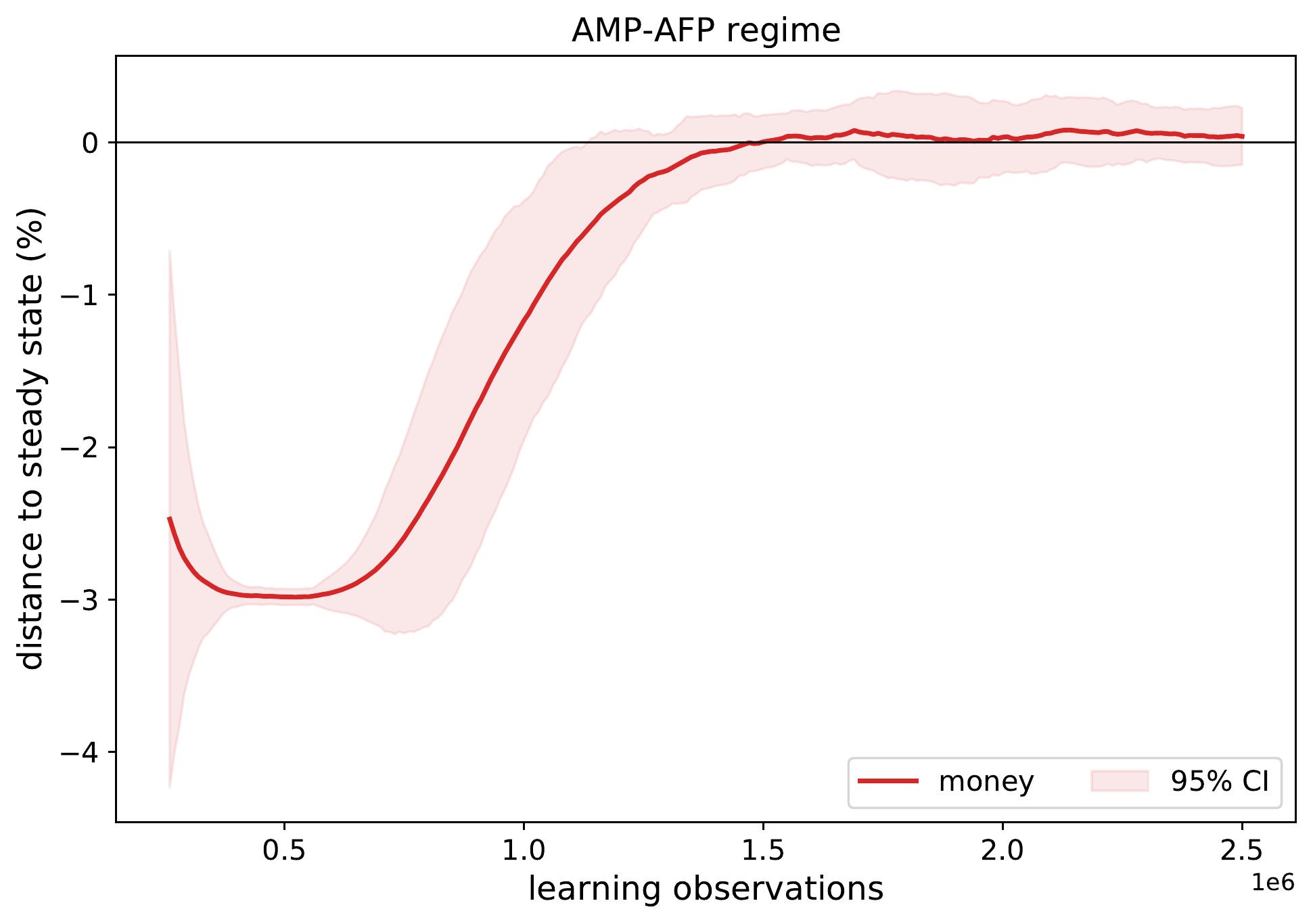}
\end{subfigure}

\begin{subfigure}{0.55\textwidth}
\hspace*{-1.cm}
\includegraphics[width=\linewidth]{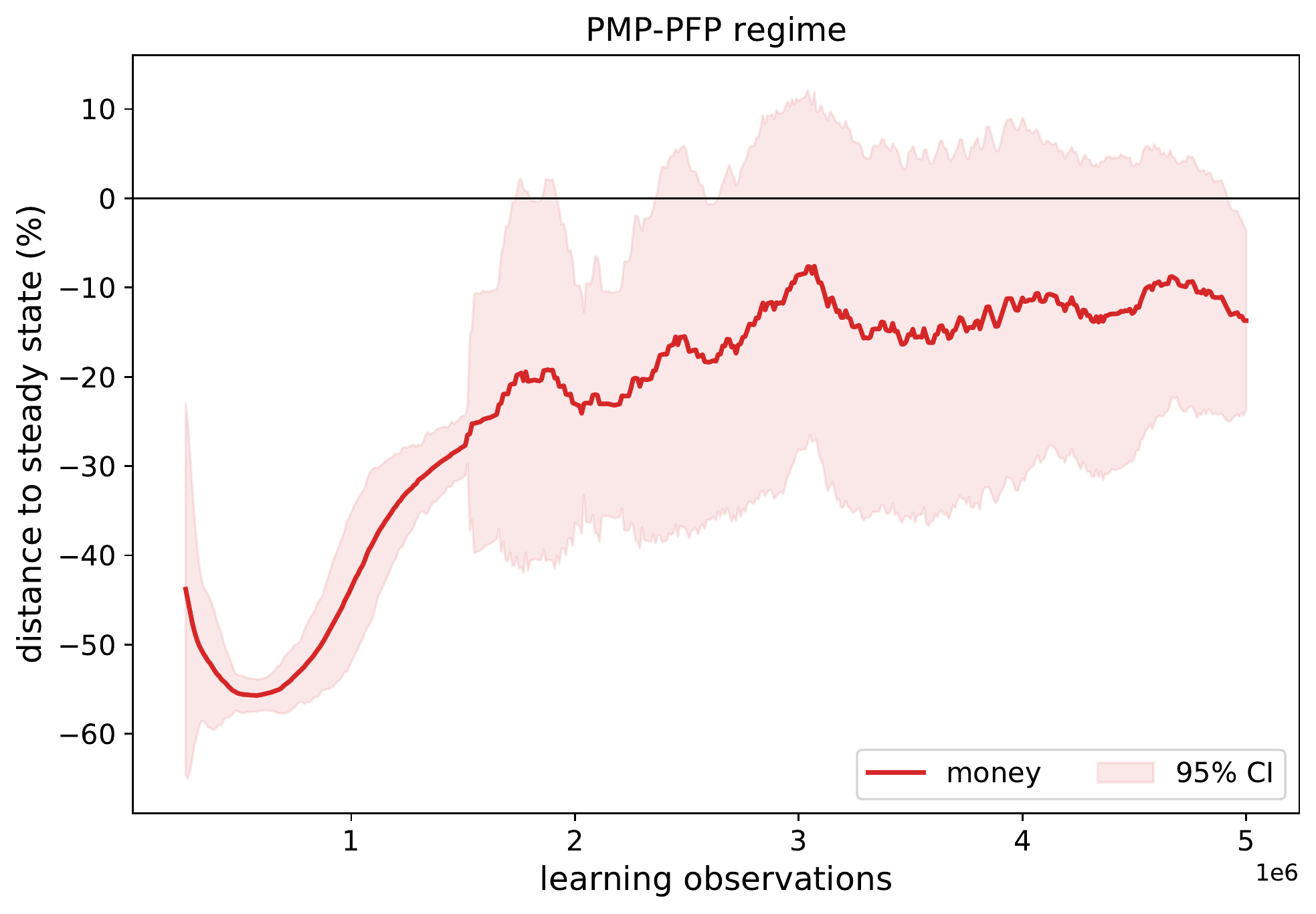}
\end{subfigure}
\begin{subfigure}{0.55\textwidth}
\hspace*{-1.cm}
\includegraphics[width=\linewidth]{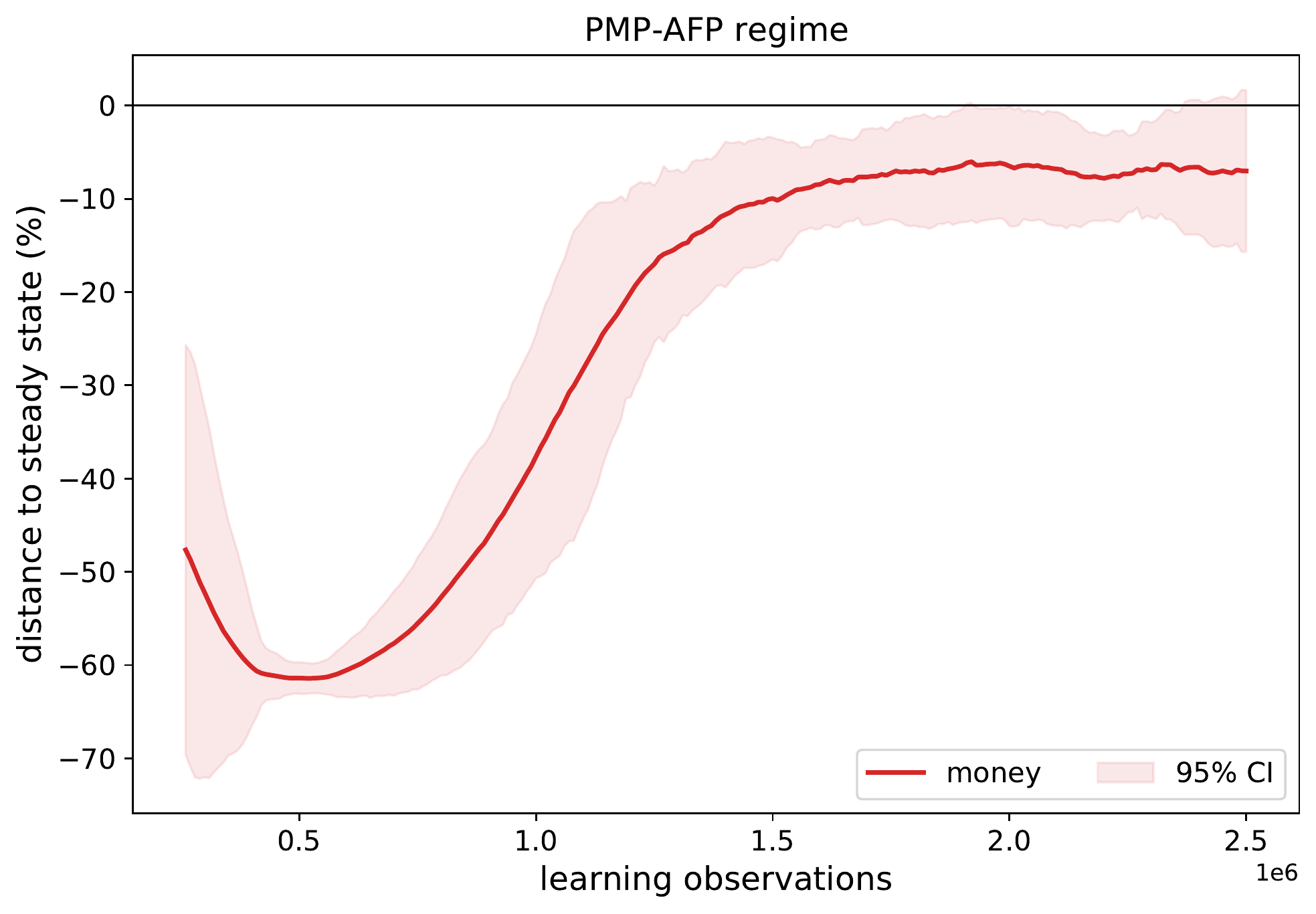}
\end{subfigure}
\caption{Comparison of steady state convergence of money holdings for different policy regimes  at end of test episodes: monetary policy (rows) and fiscal policy (columns). Vertical axis shows percentage of respective steady state value. Shaded areas show 95\% confidence intervals. Source: Authors' calculations.} \label{fig:ss_money}
\end{figure}

\newpage
\begin{figure}[b]  
\begin{subfigure}{0.55\textwidth}
\hspace*{-1.cm}
\includegraphics[width=\linewidth]{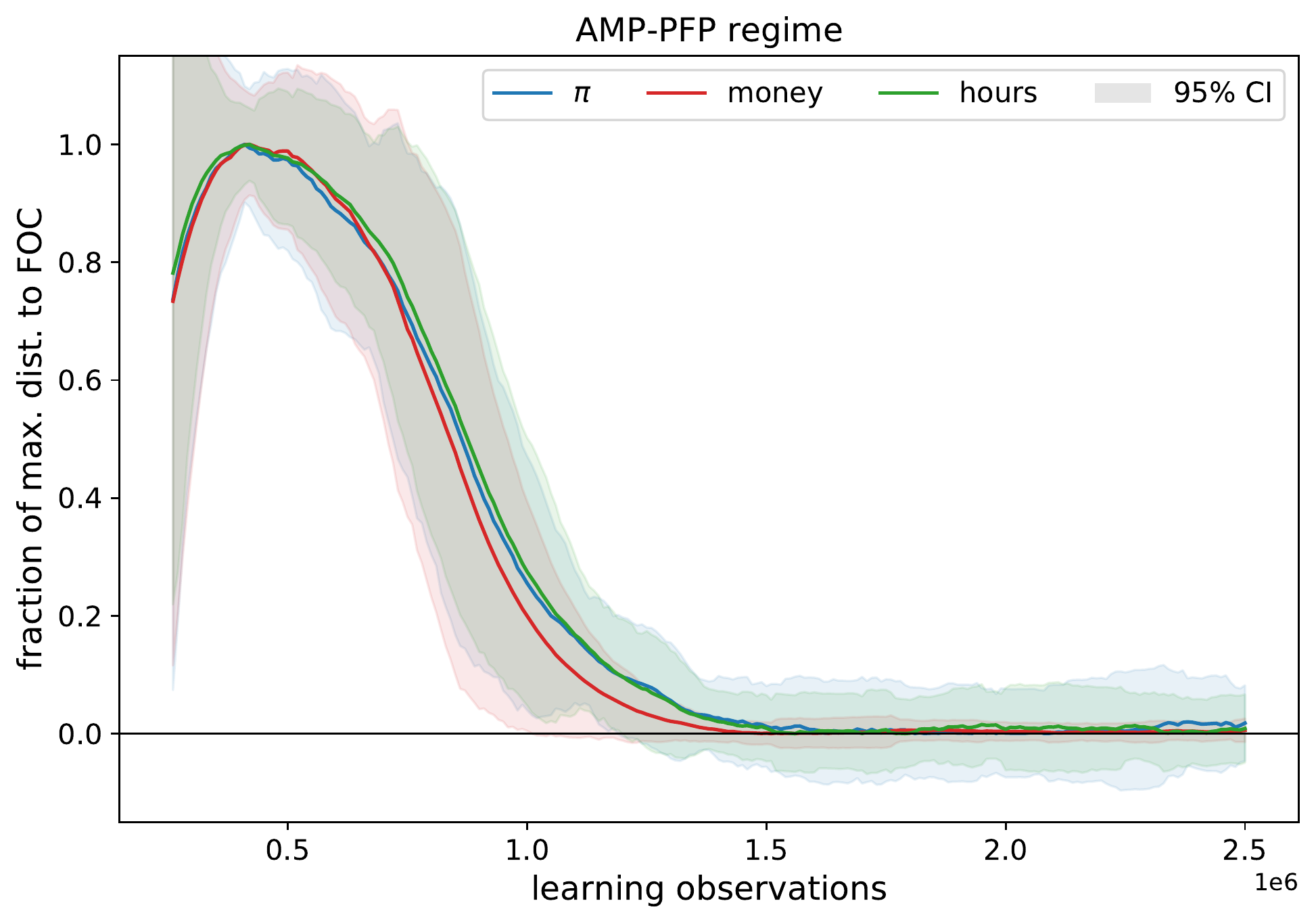}
\end{subfigure}
\begin{subfigure}{0.55\textwidth}
\hspace*{-1.cm}
\includegraphics[width=\linewidth]{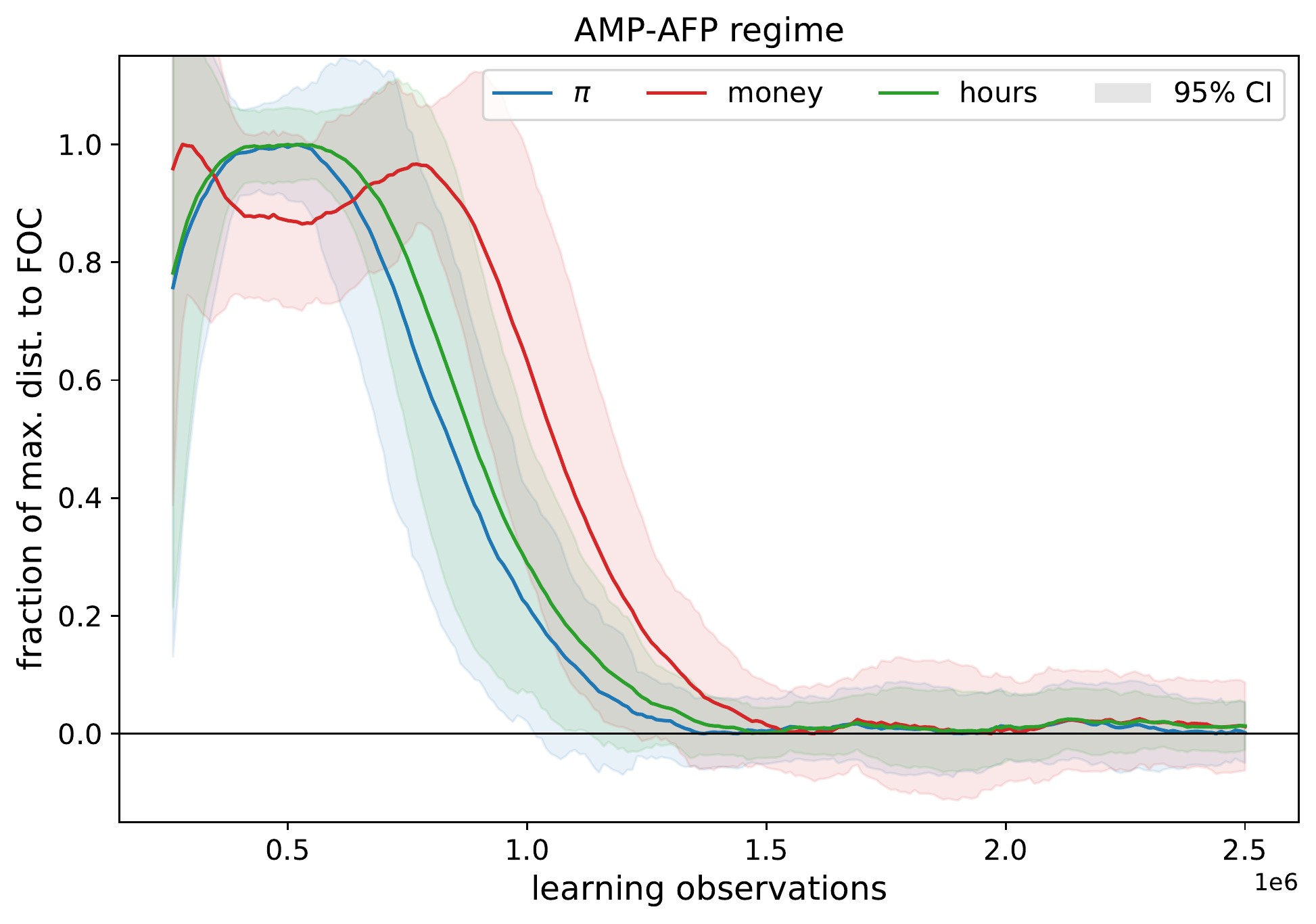}
\end{subfigure}

\begin{subfigure}{0.55\textwidth}
\hspace*{-1.cm}
\includegraphics[width=\linewidth]{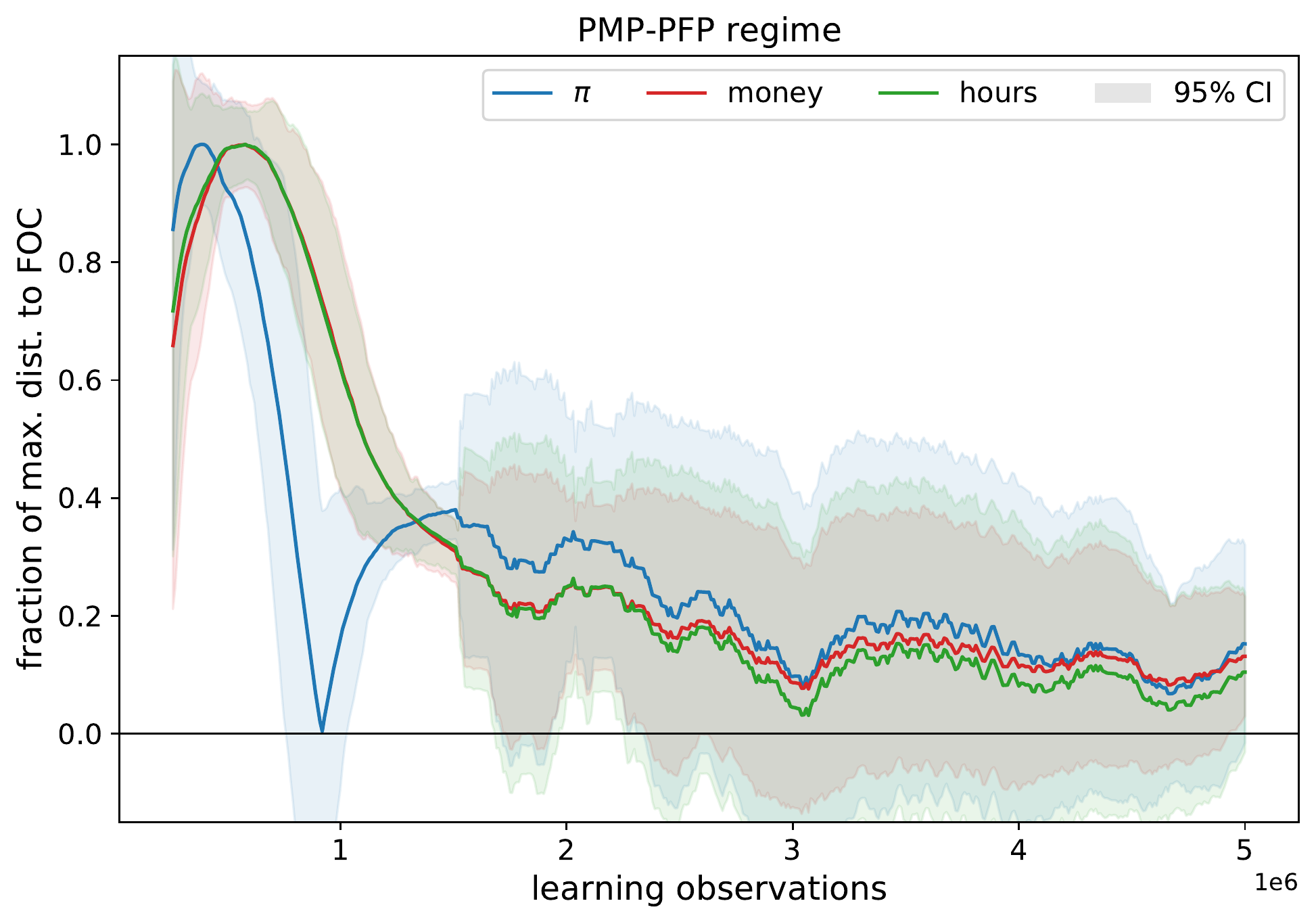}
\end{subfigure}
\begin{subfigure}{0.55\textwidth}
\hspace*{-1.cm}
\includegraphics[width=\linewidth]{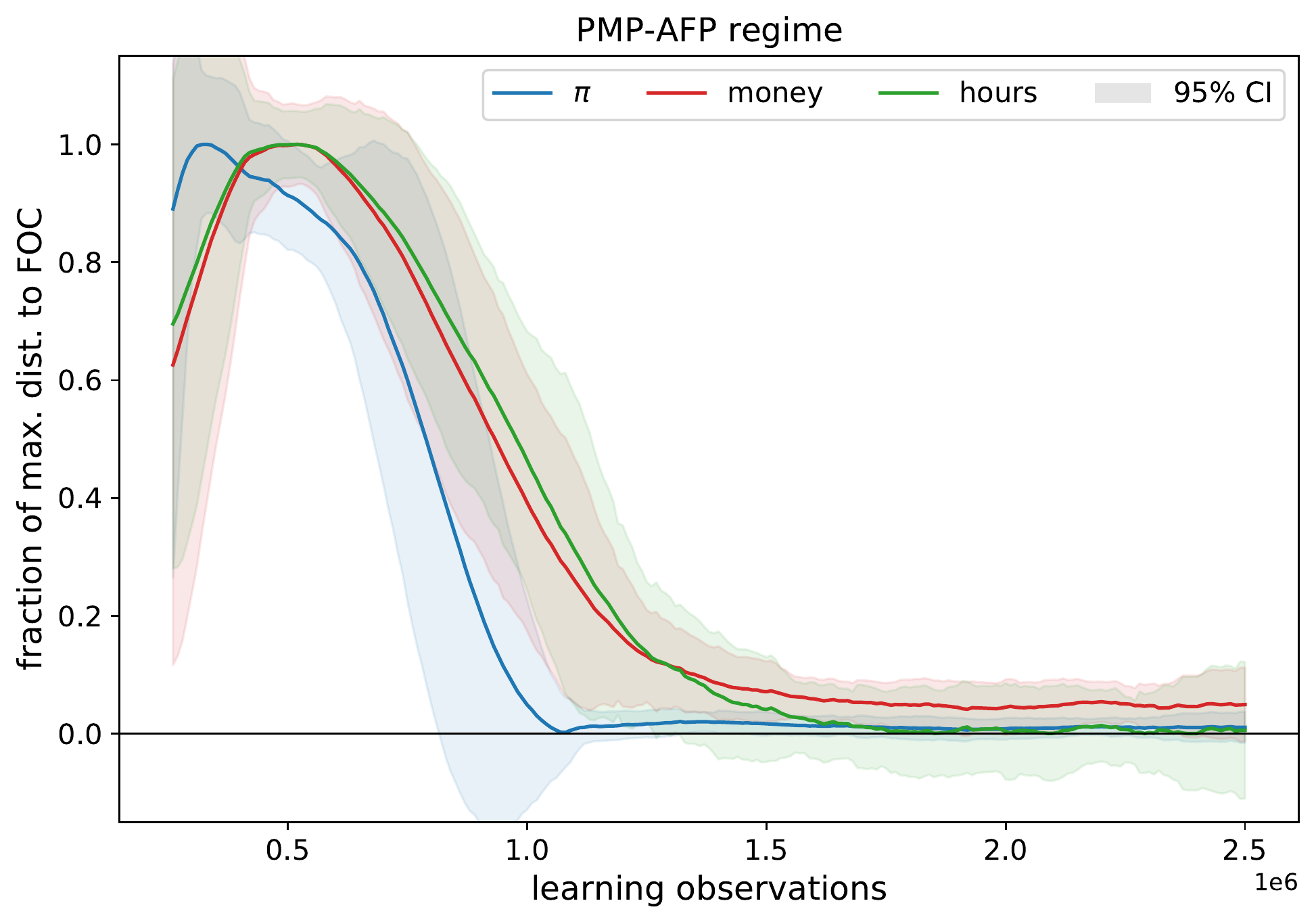}
\end{subfigure}
\caption{Comparison of FOC-learning for different policy regimes at end of test episodes: monetary policy (rows) and fiscal policy (columns). Shaded areas show 95\% confidence intervals. Source: Authors' calculations.} \label{fig:foc_all}
\end{figure}

\end{document}